\documentclass[journal]{IEEEtran}

\IEEEoverridecommandlockouts
\makeatletter
\def\ps@headings{%
\def\@oddhead{\mbox{}\scriptsize\rightmark \hfil \thepage}%
\def\@evenhead{\scriptsize\thepage \hfil \leftmark\mbox{}}%
\def\@oddfoot{}%
\def\@evenfoot{}}
\makeatother
\usepackage{subfigure}
\usepackage{colortbl}
\usepackage{bm}

\usepackage{graphicx}  
\usepackage{url}       

\usepackage{amsmath}   
\usepackage{cite}

\usepackage{amsfonts,amssymb}
\usepackage{cite}
\usepackage{stfloats}

\usepackage{diagbox}

\usepackage{cases}
\usepackage{algorithm,algorithmic}
\usepackage{multirow}
\usepackage{stfloats}
\usepackage{epstopdf}

\usepackage{xcolor}
\usepackage{xpatch}
\makeatletter
\def\changeBibColor#1{%
	\in@{#1}{
	}
	\ifin@\color{blue}\else\normalcolor\fi
}
\xpatchcmd\@bibitem
{\item}
{\changeBibColor{#1}\item}
{}{\fail}
\xpatchcmd\@lbibitem
{\item}
{\changeBibColor{#2}\item}
{}{\fail}

\makeatletter

\newcommand{\Rmnum}[1]{\expandafter\@slowromancap\romannumeral #1@}
\makeatother

\UseRawInputEncoding

\newtheorem{theorem}{{Theorem}}
\newtheorem{lemma}{{Lemma}}

\newtheorem{proposition}{{Proposition}}

\newtheorem{remark}{{Remark}}

\newcommand{\ls}[1]
    {\dimen0=\fontdimen6\the\font
     \lineskip=#1\dimen0
     \advance\lineskip.5\fontdimen5\the\font
     \advance\lineskip-\dimen0
     \lineskiplimit=.9\lineskip
     \baselineskip=\lineskip
     \advance\baselineskip\dimen0
     \normallineskip\lineskip
     \normallineskiplimit\lineskiplimit
     \normalbaselineskip\baselineskip
     \ignorespaces
    }

\begin{document}

\title{ 
	 Anti-Jamming Optimization for EM-Compliant Active RIS via Decoupling Architecture
 
   }
\vspace{10pt}


\author{\IEEEauthorblockN{Yang Cao, \emph{Graduate Student Member}, \emph{IEEE}, Wenchi Cheng, \emph{Senior Member}, \emph{IEEE}, Jingqing Wang, \emph{Member}, \emph{IEEE}, and Lifeng Wang}\\[0.2cm]
	\vspace{-5pt}

	\vspace{-25pt}
	
	\thanks{
		
		
		Yang Cao, Wenchi Cheng and Jingqing Wang are with the State Key Laboratory of Integrated Services Networks, Xidian University, Xi'an,
		710071, China (e-mails: caoyang@stu.xidian.edu.cn, wccheng@xidian.edu.cn, wangjingqing00@gmail.com).
		
		Lifeng Wang is with National Key Laboratory of Multi-domain Data Collaborative Processing and Control, Beijing, China (e-mail: wolfbox@163.com).
		
	}
}

\maketitle

\begin{abstract}
  
 Wireless communication systems are increasingly vulnerable to sophisticated jamming attacks with the rapid evolution of jamming technologies and advanced signal processing techniques. While traditional anti-jamming techniques offer limited performance gains, active reconfigurable intelligent surfaces (RISs) have emerged as a promising channel-domain solution for improving resilience against jamming. Nonetheless, existing studies often rely on simplified electromagnetic (EM) models that do not fully capture mutual coupling (MC) and impedance mismatches in RIS hardware. In this paper, we propose an EM-compliant active (EMC-Active) RIS model for anti-jamming systems, explicitly incorporating the EM and physical properties at active RIS, such as MC effects, channel correlation, and discrete phase. To evaluate the anti-jamming performance of the proposed EMC-Active RIS, we develop a low-complexity alternating optimization (AO) algorithm based on the decoupling architecture (DA) to maximize the ergodic achievable rate. By leveraging the DA to explicitly eliminate MC effects among REs, the original coupled system is transformed into a tractable and scalable uncoupled representation. Numerical results demonstrate that the DA-based AO algorithm can significantly reduce the modeling and optimization complexity and efficiently solve the problem in an alternating manner with substantially reduced iteration overhead.

\end{abstract}

\vspace{-5pt}


\vspace{-15pt}
\section{Introduction}
 \IEEEPARstart{O}{wing} to the superposition and broadcast nature of electromagnetic (EM) wave propagation, wireless communication systems are intrinsically susceptible to intentional jamming from malicious jammers. With the rapid advancement of hardware and signal processing capabilities, jamming devices have evolved from conventional non-adaptive schemes---such as single-tone, frequency-sweeping, or tracking-based jamming---toward more sophisticated and adaptive jamming strategies capable of exploiting multiple jamming patterns \cite{Cao2}. In parallel, anti-jamming techniques have progressed from traditional spread-spectrum approaches to more advanced intelligent countermeasures. Nevertheless, as jamming mechanisms continue to gain intelligence, existing anti-jamming methods that primarily rely on power control, frequency agility, or spatial processing are facing diminishing performance returns. In particular, both jammers and legitimate systems increasingly employ similar intelligent architectures and learning-based algorithms, which significantly reduces the likelihood of achieving a decisive advantage on either side and ultimately leads to a persistent confrontation stalemate \cite{yao2023wireless}.
 
 To counter increasingly sophisticated intelligent jamming attacks, there is a growing demand for anti-jamming solutions in new design dimensions. Enabled by recent advances in metamaterials, reconfigurable intelligent surfaces (RISs) have emerged as a promising and energy-efficient technology for reshaping wireless propagation environments \cite{Cao2}. By adaptively adjusting the amplitude or phase of incident EM waves through reflection elements (REs), RISs introduce a new degree of freedom in the channel domain, offering fresh opportunities for anti-jamming system design. Recent studies have explored RIS-assisted robust beamforming for secure communication under jamming and eavesdropping \cite{Sun1}, as well as beamforming strategies for multiple RISs-assisted anti-jamming system \cite{Cao1}. Despite their advantages, conventional passive RISs suffer from an inherent limitation known as double-fading attenuation, where the cascaded RIS-assisted link suffers two large-scale fading effects, significantly constraining the achievable performance gains. To address this limitation, the concept of active RIS has been introduced \cite{Active_RIS1}, where each RE is equipped with a a reflection amplifier (RA) and a phase shifter (PS)  capable of jointly controlling the amplitude and phase of incident signals. Recent studies indicate that active RISs can deliver significant performance gains compared to passive designs \cite{Active_RIS2}. In the context of secure communications, active RIS-assisted anti-jamming systems have demonstrated substantial improvements even with a limited number of REs \cite{Cao2}. Furthermore, an effective optimization method has been developed to address the non-convex security rate maximization problem arising in active RIS-aided secure scenarios \cite{Aided-Secure-Tran}.
 
 
 Most existing passive RIS-assisted communication studies model RIS elements as ideal scatterers, which often deviates from their actual EM properties and physical realizations, especially when mutual coupling (MC) and impedance mismatch are present. To address this issue, recent works have introduced EM-compliant RIS models based on multiport network theory using either impedance ($Z$-) \cite{Renzo_2021} or scattering ($S$-) \cite{Shen_2022} parameters, enabling explicit characterization of MC effects and hardware constraints and revealing their significant impact on system-level optimization \cite{Z_parameter,S_parameter,thus3}. However, similar limitations persist for active RISs, as most existing studies focus on system-level optimization while relying on simplified EM-incompatible models \cite{Active_RIS1,Active_RIS2}. Although preliminary efforts have investigated the EM design of REs at active RIS \cite{Active_RIS_cankao,Active_RIS_cankao1}, an EM-compliant modeling framework that connects active RIS hardware characteristics with system-level optimization is still lacking. To fill this gap, our prior work \cite{Cao3} proposed an EM-compliant active (EMC-Active) RIS model for MIMO systems that explicitly incorporates MC and accurately captures the EM properties of active RIS hardware.

 However, the design of EMC-Active RIS-assisted anti-jamming systems has not yet been thoroughly explored. Existing EM-compliant RIS models \cite{Cao3} are primarily tailored to simple communication scenarios and cannot be directly extended to anti-jamming systems, which require robustness against unknown jammers and optimization based on imperfect channel state information (CSI) \cite{Sun1}. Thus, these models need to be extended to suit the requirements of anti-jamming applications. Moreover, MC-aware optimization methods usually rely on computationally intensive iterative algorithms, whose complexity grows rapidly with the number of RIS elements \cite{BD_RIS}. Motivated by these challenges, in this paper we develop an EMC-Active RIS-assisted anti-jamming framework that jointly captures the EM and physical properties of active RIS, such as mutual coupling effects, channel correlation, and discrete phase. To evaluate the anti-jamming performance of the proposed EMC-Active RIS, we develop a low-complexity alternating optimization (AO) algorithm based on the decoupling architecture (DA) to maximize the ergodic achievable rate. By leveraging the DA to explicitly eliminate MC effects among REs, the original coupled system is transformed into an equivalent uncoupled representation. This transformation alters the solvability structure of the original optimization problem, enabling the originally intractable coupled optimization to be transformed into a manageable and scalable form. Finally, the numerical results show that the DA-based AO algorithm can significantly reduce the modeling and optimization complexity, which can efficiently solve the problem in an alternating manner with substantially reduced iteration overhead.

 The remainder of this paper is organized as follows. Section~\ref{sec:System_model} presents the EMC-Active RIS-assisted anti-jamming system model and develops its corresponding $S$-parameter and $Z$-parameter representations. Section~\ref{sec:Decoupling_Problem_formulation} derives the decoupling architecture and formulates the associated non-convex optimization problem. In Section~\ref{sec:Alternating_optimization_perfect_CSI}, the DA-based AO algorithm is proposed to solve the formulated problem. Section~\ref{sec:simulation} provides numerical results to validate the effectiveness of the proposed framework. Finally, conclusions are drawn in Section~\ref{sec:Conclusion}.

 \begin{figure}[t]
	\vspace{-15pt}
	\centering
	\includegraphics[scale=0.10]{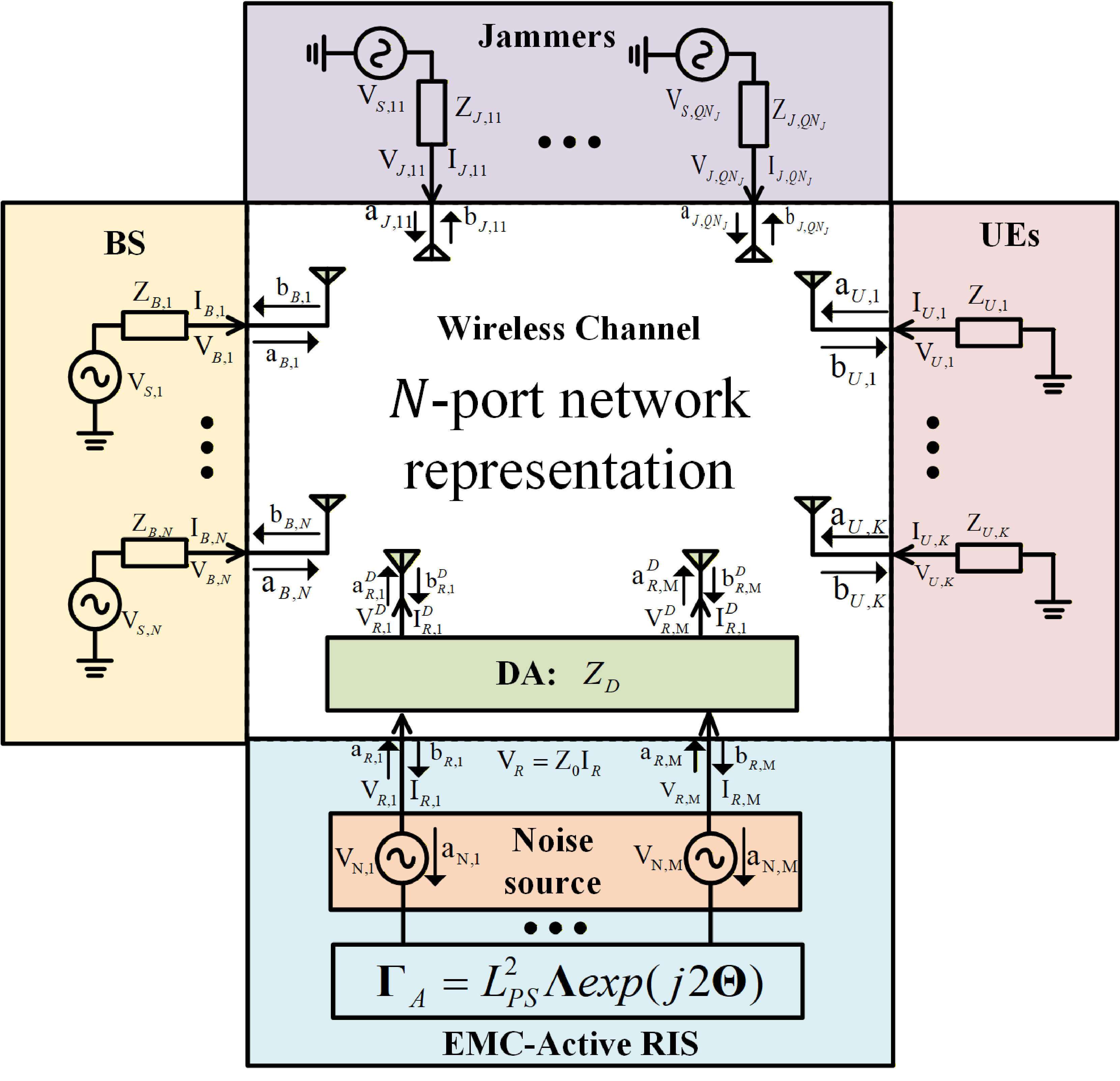}
	\vspace{-10pt}
	\caption{Multiport model of EMC-Active RIS with a decoupling architecture.}
	\vspace{-13pt}
	\label{fig:shiyitu}
\end{figure}

\vspace{-5mm}

 \section{System Model}
 \label{sec:System_model}
 
 We consider an active RIS-assisted anti-jamming system, comprising a BS with $N$ antennas, $Q$ malicious jammers with $N_J$ antennas, $K$ single-antenna UEs, and an active RIS with $M$ REs. The BS, equipped with a uniform linear array (ULA) \cite{OAM}, serves the UEs with the support of active RIS, while jammers equipped with ULA send malicious jamming signals to the UE seeking to disrupt legitimate communications, where REs are configured as a uniform planar array (UPA) with $M_h$ rows and $M_v$ columns (i.e., $M = M_hM_v$). Let $\mathcal{M}=\{1,2,\cdots,M\}$, $\mathcal{K}=\{1,2,\cdots,K\}$, and $\mathcal{Q}=\{1,2,\cdots,Q\}$, denote the RE set, the UE set, and the jammer set, respectively.
 
 \subsection{EM-Compliant Active RIS Model}
 \label{subsec:EM_Active_RIS}
 
 Active RIS is equipped with reflection amplifiers (RAs) for signal amplification and phase shifters (PSs) for phase adjustment, effectively eliminating the double-fading attenuation effect introduced by traditional passive RIS. To accurately characterize the impact of the EM and physical properties of active RIS architectures in the system, we introduce an EM-compliant active RIS (abbreviated as EMC-Active RIS) model proposed in our previous work \cite{Cao3}, as shown in Fig.~\ref{fig:shiyitu}. Thus, the reflection coefficient matrix $\boldsymbol{\rm \Gamma}_A$ of EMC-Active RIS can be written as 
 
 \vspace{-5mm}
 \begin{small}
 	 \begin{align}\label{gamma_0}
 	\boldsymbol{\rm \Gamma}_A=L_{P\hspace{-0.3mm}S}^{2}\boldsymbol{\rm \Lambda}\exp\left(j2\boldsymbol{\rm \Theta}\right),
 	\end{align}
 \end{small}where $L_{P\hspace{-0.3mm}S}$, $\boldsymbol{\rm \Lambda}=\mbox{Diag}\left(\left[\alpha_{1},\cdots,\alpha_{M}\right]^T\right)$, and $\boldsymbol{\rm \Theta}=\mbox{Diag}\left(\left[\theta_1,\cdots,\theta_M\right]^T\right)$ denote the insertion loss of the PS, the amplitude matrix, and the phase shift matrix, respectively, where $\alpha_{m}$, $\theta_m$, and $\mbox{Diag}\left(\cdot\right)$ represent the amplitude of RA, the phase of PS at the $m$th RE, and a diagonal matrix respectively.
 
 Based on the EMC-Active RIS model and multiport network theory \cite{Ivrlac_Nossek_2010}, we can further develop an EMC-Active RIS-assisted anti-jamming communication system model, which further incorporates the EM properties at the active RIS---primarily the EM mutual coupling (MC) among the REs---on top of the conventional communication model. Specifically, to begin with, we define the transmitted signals at the BS and the $q$th jammer as follows:
 
 \vspace{-5mm}
 \begin{small}
 	\begin{align}
 	\boldsymbol{\rm x}=\sum_{k=1}^K\boldsymbol{\rm w}_ks_k, ~~~~
 	\boldsymbol{\rm {x}}_{J,q} = \sum_{k=1}^K\boldsymbol{\rm {w}}_{J,qk}s_{J,qk}, \forall q\in \mathcal{Q},
 	\end{align}
 \end{small}where $\boldsymbol{\rm {w}}_k \in \mathbb{C}^{N\times 1}$, $s_k\sim\mathcal{CN}(0,1)$, $\boldsymbol{\rm {w}}_{J,qk} \in \mathbb{C}^{N_J\times 1}$, and $s_{J,qk}\sim\mathcal{CN}(0,1)$ denote the transmit beamforming vector and symbol at the BS and the $q$th jammer for the $k$th UE, respectively. Besides, we define that $\sum_{k=1}^K\Vert\boldsymbol{\rm {w}}_k\Vert^2\le P_{\max}$ and $P_{J,q}=\sum_{k=1}^K\Vert\boldsymbol{\rm {w}}_{J,qk}\Vert^2\le P_{J,\max}, \forall q\in \mathcal{Q}$, where $P_{\max}$ and $P_{J,\max}$ are the maximum transmit power budget at the BS and jammers, respectively. The notation $\mathcal{CN}(0,1)$ denotes the a zero mean and unit variance random variable with independent and identically distributed (i.i.d.).
 
 According to Fig.~\ref{fig:shiyitu}, based on the scattering parameter ($S$-parameter) representation, the signal received by the $k$th UE in the  EMC-Active RIS-assisted anti-jamming system can be formulated as follows \cite{Cao3}:
 
 \vspace{-5mm}
 \begin{small}
    \begin{align}\label{received_signal}
 	\begin{split}
 	y_k=\boldsymbol{\rm H}^{E}_{S,k}\boldsymbol{\rm x}+\sum_{q=1}^Q\boldsymbol{\rm H}^{J}_{S,qk}\boldsymbol{\rm x}_{J,q}+\boldsymbol{\rm H}^{N}_{S,k}\boldsymbol{\rm n}_{R}+{n}_{k}, \forall k\in \mathcal{K},
 	\end{split}
 	\end{align}
 \end{small}where we define

\vspace{-3mm}
\begin{small}
	\begin{align}\label{HE}
	&\boldsymbol{\rm H}^{E}_{S,k}\triangleq\boldsymbol{\rm S}_{BU,k}+\boldsymbol{\rm  S}_{RU,k}\left(\boldsymbol{\rm I}_{M}-\boldsymbol{\rm \Gamma}_{A}\boldsymbol{\rm  S}_{AA}\right)^{-1}\boldsymbol{\rm \Gamma}_{A}\boldsymbol{\rm S}_{BR},\\
	&\boldsymbol{\rm H}^{J}_{S,qk}\triangleq\boldsymbol{\rm  S}_{JU,qk}+\boldsymbol{\rm  S}_{RU,k}\left(\boldsymbol{\rm I}_{M}-\boldsymbol{\rm \Gamma}_{A}\boldsymbol{\rm  S}_{AA}\right)^{-1}\boldsymbol{\rm \Gamma}_{A}\boldsymbol{\rm  S}_{JR,q},\label{HJ}\\
	&\boldsymbol{\rm H}^{N}_{S,k}\triangleq\boldsymbol{\rm  S}_{RU,k}\left(\boldsymbol{\rm I}_{M}-\boldsymbol{\rm \Gamma}_{A}\boldsymbol{\rm  S}_{AA}\right)^{-1}\boldsymbol{\rm \Gamma}_{A},\label{HN}
	\end{align}
\end{small}where $\boldsymbol{\rm S}_{\varpi_1\varpi_2}$ for $\varpi_1,\varpi_2\in\left\{B,J,R,U\right\}$ denotes the $S$-parameter representation of the channel matrix between the end of $\varpi_1$ and $\varpi_2$ with $\left\{B,J,R,U\right\}$ referring to the BS, jammer, EMC-active RIS, and UE, respectively. The notations $\boldsymbol{\rm n}_{ R}\sim\mathcal{CN}(0,\sigma^2_{ R}\boldsymbol{\rm I}_M)$ and $n_k\sim\mathcal{CN}(0,\sigma_k^2)$ denote the additive complex Gaussian noises at the EMC-Active RIS and UEs, respectively. The notation $\boldsymbol{\rm I}$ is the identity matrix of the dimension corresponding to the subscripts. The notation $\boldsymbol{\rm  S}_{AA}$ represents the scattering matrix at the EMC-Active RIS, where the diagonal elements indicate the self-impedance of the REs, while the off-diagonal elements indicate the EM MC among REs. In addition, the output signal, denoted by $\boldsymbol{\rm y}_A$, at the EMC-Active RIS is formulated as follows:

\vspace{-3mm}
\begin{small}
 \begin{align}\label{PA1}
	\begin{split}
	\boldsymbol{\rm y}_A =\boldsymbol{\rm \breve {H}}_{S}^{E}\boldsymbol{\rm x}+\sum_{q=1}^Q\boldsymbol{\rm \breve H}^{J}_{S,q}\boldsymbol{\rm x}_{J,q}+\boldsymbol{\rm \breve H}^{N}_{S}\boldsymbol{\rm n}_{R},
	\end{split}
	\end{align}
\end{small}where we introduce the definition as follows:

\vspace{-3mm}
\begin{small}
	\begin{align}\label{H_A1}
	&\boldsymbol{\rm \breve {H}}_{S}^{E}\triangleq\left(\boldsymbol{\rm I}_{M}-\boldsymbol{\rm \Gamma}_{A}\boldsymbol{\rm  S}_{AA}\right)^{-1}\boldsymbol{\rm \Gamma}_{A}\boldsymbol{\rm S}_{BR},\\
	&\boldsymbol{\rm \breve {H}}_{S,q}^{J}\triangleq\left(\boldsymbol{\rm I}_{M}-\boldsymbol{\rm \Gamma}_{A}\boldsymbol{\rm  S}_{AA}\right)^{-1}\boldsymbol{\rm \Gamma}_{A}\boldsymbol{\rm S}_{JR,q},\\
	&\boldsymbol{\rm \breve {H}}_{S}^{N}\triangleq\left(\boldsymbol{\rm I}_{M}-\boldsymbol{\rm \Gamma}_{A}\boldsymbol{\rm  S}_{AA}\right)^{-1}\boldsymbol{\rm \Gamma}_{A}.\label{H_A2}
	\end{align}
\end{small}
 It is noteworthy that the proposed EMC-Active RIS-assisted anti-jamming system model establishes a clear and critical bridge between the MC property inherent in densely radiating REs and the performance of the anti-jamming system, which facilitates a deeper understanding of the roles played by MC within the system. However, in previous related work \cite{Z_parameter,S_parameter,thus3}, we observed that despite the fact that EMC-Active RIS-assisted anti-jamming system model is sufficiently concise to facilitate system-level optimization, the matrix inversion and complex coupling of optimization variables introduced by MC in the equivalent channels $\boldsymbol{\rm H}^{E}_{S,k}$, $\boldsymbol{\rm H}^{J}_{S,qk}$, and $\boldsymbol{\rm H}^{N}_{S,k}$ still entail extremely high computational costs.
 
 \vspace{-3mm}
 \subsection{Parameter Representation Transformation}
 
 To seek more efficient optimization methods, we first need to perform a parameter representation transformation on the EMC-Active RIS-assisted anti-jamming system model. The $S$-parameter characterized system model is more compatible with traditional active RIS-assisted system models, while the $Z$-parameter characterized system model explicitly captures the circuit topology of the MC introduced by densely radiating REs. Mathematically, the two models are equivalently convertible \cite{BD_RIS}. Therefore, we derive the $Z$-parameter representation of the EMC-Active RIS-assisted anti-jamming system in the following theorem.
 
 \begin{theorem}\label{theorem_Z}
 	Based on the equivalence relationship between $Z$-parameters and $S$-parameters \cite[Result 2]{BD_RIS}, we derive the $Z$-parameter representation of the signal received by the $k$th UE in the EMC-Active RIS-assisted anti-jamming system, as follows:
 	
 	\vspace{-3mm}
 	\begin{small}
 		\begin{align}\label{Z_model1}
 		\begin{split}
 		y_k\hspace{-0.5mm}=\boldsymbol{\rm H}^{E}_{Z,k}\boldsymbol{\rm x}\hspace{-0.5mm}+\hspace{-1mm}\sum_{q=1}^Q\boldsymbol{\rm H}^{J}_{Z,qk}\boldsymbol{\rm x}_{J,q}\hspace{-0.5mm}+\boldsymbol{\rm H}^{N}_{Z,k}\boldsymbol{\rm n}_{R}+{n}_{k}, \forall k\in \mathcal{K},
 		\end{split}
 		\end{align}
 	\end{small}where we introduce the definitions as follows:
 	
 	\vspace{-3mm}
 	\begin{small}
 		\begin{align}\label{HZE}
 		&\boldsymbol{\rm H}^{E}_{Z,k}\triangleq\frac{1}{2Z_0}\left(\boldsymbol{\rm Z}_{BU,k}-\boldsymbol{\rm  Z}_{RU,k}\left(\boldsymbol{\rm Z}_{A}+\boldsymbol{\rm  Z}_{AA}\right)^{-1}\boldsymbol{\rm Z}_{BR}\right),\\
 		&\boldsymbol{\rm H}^{J}_{Z,qk}\hspace{-1mm}\triangleq\hspace{-1mm}\frac{1}{2Z_0}\hspace{-1mm}\left(\boldsymbol{\rm  Z}_{JU,qk}-\boldsymbol{\rm  Z}_{RU,k}\left(\boldsymbol{\rm Z}_{A}+\boldsymbol{\rm  Z}_{AA}\right)^{-1}\boldsymbol{\rm  Z}_{JR,q}\right)\hspace{-1mm},\label{HZJ}\\
 		&\boldsymbol{\rm H}^{N}_{Z,k}\triangleq\frac{1}{2Z_0}\boldsymbol{\rm  Z}_{RU,k}\left(\boldsymbol{\rm I}_{M}-\left(\boldsymbol{\rm Z}_{A}+\boldsymbol{\rm  Z}_{AA}\right)^{-1}\boldsymbol{\rm  Z}_{AA}^+\right),\label{HZN}
 		\end{align}
 	\end{small}where $\boldsymbol{\rm  Z}_{AA}^+ = \boldsymbol{\rm  Z}_{AA}+Z_0\boldsymbol{\rm  I}_{M}$ and $\boldsymbol{\rm Z}_{\varpi_1\varpi_2}$ for $\varpi_1,\varpi_2\in\left\{B,J,R,U\right\}$ denotes the $Z$-parameter representation of the channel matrix between the end of $\varpi_1$ and $\varpi_2$. The notations $\boldsymbol{\rm  Z}_{A}$ and $\boldsymbol{\rm  Z}_{AA}$ represent the tunable load impendances and the impendance matrix  characterizing the MC and self-impedance at the EMC-Active RIS, respectively. Similarly, we derive the $Z$-parameter representation of the output signal at the EMC-Active RIS, as follows:
 
 \vspace{-3mm}
 \begin{small}
 	\begin{align}\label{PAZ1}
 	\begin{split}
 	\boldsymbol{\rm y}_A =\boldsymbol{\rm \breve {H}}_{Z}^{E}\boldsymbol{\rm x}+\sum_{q=1}^Q\boldsymbol{\rm \breve H}^{J}_{Z,q}\boldsymbol{\rm x}_{J,q}+\boldsymbol{\rm \breve H}^{N}_{Z}\boldsymbol{\rm n}_{R},
 	\end{split}
 	\end{align}
 \end{small}where we introduce the definition as follows:

\vspace{-3mm}
  \begin{small}
	\begin{align}\label{H_ZA1}
	&\boldsymbol{\rm \breve {H}}_{Z}^{E}\triangleq\frac{1}{2Z_0}\boldsymbol{\rm Z}_{BR}-\frac{1}{2Z_0}\left(\boldsymbol{\rm Z}_{A}+\boldsymbol{\rm  Z}_{AA}\right)^{-1}\boldsymbol{\rm Z}_{AA}^+\boldsymbol{\rm Z}_{BR},\\
	&\boldsymbol{\rm \breve {H}}_{Z,q}^{J}\triangleq\frac{1}{2Z_0}\boldsymbol{\rm Z}_{JR,q}-\frac{1}{2Z_0}\left(\boldsymbol{\rm Z}_{A}+\boldsymbol{\rm  Z}_{AA}\right)^{-1}\boldsymbol{\rm Z}_{AA}^+\boldsymbol{\rm Z}_{JR,q},\\
	&\boldsymbol{\rm \breve {H}}_{Z}^{N}\triangleq\frac{1}{2Z_0}\boldsymbol{\rm Z}_{AA}^+-\frac{1}{2Z_0}\boldsymbol{\rm Z}_{AA}^+\left(\boldsymbol{\rm Z}_{A}+\boldsymbol{\rm  Z}_{AA}\right)^{-1}\boldsymbol{\rm Z}_{AA}^+.\label{H_ZA2}
	\end{align}	
  \end{small}
 	\end{theorem}
 \textit{Proof:} Please see Appendix A.$\hfill\blacksquare$
 
  \begin{remark}
 	Inspired by the $S$-parameter based EMC-Active RIS assisted MIMO system model derived in \cite{Cao3}, we develop the first $Z$-parameter EMC-Active RIS model for active RIS-assisted anti-jamming systems. This model provides an explicit representation of the circuit topology for MC effect in active RIS, allowing us to analyze the physical and EM mechanisms of MC introduced by densely deployed REs. 
 \end{remark}
 
 \vspace{-3mm}

 \section{Decoupling Architecture And Problem Formulation}
 \label{sec:Decoupling_Problem_formulation}
  
 \subsection{Decoupling Architecture}
 \label{sec:Decoupling}
  By examining models (\ref{PA1}) and (\ref{PAZ1}), it can be seen that the inherent MC effect between REs leads to deep coupling between the adjustable and MC components of the EMC-Active RIS. This coupling is further incorporated into the objective function and various constraints of the subsequent optimization problem through matrix inversion, evolving the original optimization of the EMC-Active RIS into a strongly coupled non-convex problem. To address this issue, we propose a decoupling architecture (DA) \cite{DN1} at the front end of the active RIS to eliminate the negative impact of MC among REs on the solvability of the optimization problem. In previous work, these decoupling architectures were commonly employed in transmit and receive arrays to address coupling between antenna elements. Inspired by the above, we extend this approach to the EMC-Active RIS model. 
 
 To begin with, in the EMC-Active RIS-assisted anti-jamming model, the relationship between voltage $\boldsymbol{\rm {v}}_A$ and current $\boldsymbol{\rm {i}}_A$ at RIS is as follows:
 
 \vspace{-3mm}
 \begin{small}
 	\begin{align}\label{De0}
 	\boldsymbol{\rm {v}}_A = -\boldsymbol{\rm {Z}}_A\boldsymbol{\rm {i}}_A.
 	\end{align}
 \end{small}When the decoupling architecture is loaded on the front end of the RIS, the updated relationship is

 \vspace{-3mm}
 \begin{small}
	\begin{align}
	\boldsymbol{\rm {v}}_A^D = -\boldsymbol{\rm {Z}}_A^D\boldsymbol{\rm {i}}_A^D.
	\end{align}
 \end{small}To derive the expression for $\boldsymbol{\rm {Z}}_A^D$, we represent the decoupled architecture as $\boldsymbol{\rm {Z}}_D$, which is reversible and lossless (i.e., $\boldsymbol{\rm {Z}}_D=\boldsymbol{\rm {Z}}_D^T$ and $\operatorname{Re}\left\{\boldsymbol{\rm {Z}}_D\right\}=0$), where $(\cdot)^T$ and $\operatorname{Re}\{\cdot\}$ denote the transpose and real part extraction of a matrix. We can then establish the relationship as follows:
 \begin{equation}\label{De1}\small
 \begin{bmatrix}
 \boldsymbol{\rm {v}}_A \\
 \boldsymbol{\rm {v}}_A^D
 \end{bmatrix}
 =\begin{bmatrix}
 \boldsymbol{\rm {Z}}_D(1,1) & \boldsymbol{\rm {Z}}_D(1,2) \\
 \boldsymbol{\rm {Z}}_D(2,1) & \boldsymbol{\rm {Z}}_D(2,2)
 \end{bmatrix}
 \begin{bmatrix}
 \boldsymbol{\rm {i}}_A\\
 -\boldsymbol{\rm {i}}_A^D
 \end{bmatrix}.
 \end{equation}
 Substituting Eq.~(\ref{De0}) into (\ref{De1}) and using $\boldsymbol{\rm {Z}}_D(2,1)=\boldsymbol{\rm {Z}}^T_D(1,2)$, we can get the relationship between $\boldsymbol{\rm {v}}_A^D$ and $\boldsymbol{\rm {i}}_A^D$ as
 
 \vspace{-3mm}
 \begin{small}
 	\begin{align}
 	\boldsymbol{\rm {v}}_A^D\hspace{-1mm} =\hspace{-1mm} -\hspace{-1mm}\left(\hspace{-1mm}\boldsymbol{\rm {Z}}_D(2,2)\hspace{-1mm}-\hspace{-1mm}\boldsymbol{\rm {Z}}^T_D(1,2)(\boldsymbol{\rm {Z}}_D(1,1)\hspace{-1mm}+\hspace{-1mm}\boldsymbol{\rm {Z}}_A)^{-1}\boldsymbol{\rm {Z}}_D(1,2)\hspace{-1mm}\right)\hspace{-1mm}\boldsymbol{\rm {i}}_A^D.
 	\end{align}
 \end{small}Thus, the expression of $\boldsymbol{\rm {Z}}_A^D$ can be derived as 

\vspace{-3mm}
\begin{small}
	\begin{align}\label{ZAD_1}
	\boldsymbol{\rm {Z}}_A^D =\boldsymbol{\rm {Z}}_D(2,2)\hspace{-1mm}-\hspace{-1mm}\boldsymbol{\rm {Z}}^T_D(1,2)(\boldsymbol{\rm {Z}}_D(1,1)\hspace{-1mm}+\hspace{-1mm}\boldsymbol{\rm {Z}}_A)^{-1}\boldsymbol{\rm {Z}}_D(1,2),
	\end{align}
\end{small}The DA transforms $\boldsymbol{\rm {Z}}_A$ into $\boldsymbol{\rm {Z}}_A^D$. Therefore, by substituting $\boldsymbol{\rm {Z}}_A$ with $\boldsymbol{\rm {Z}}_A^D$ in $\boldsymbol{\rm H}^{E}_{Z,k}$, $\boldsymbol{\rm H}^{J}_{Z,qk}$, $\boldsymbol{\rm H}^{N}_{Z,k}$, $\boldsymbol{\rm \breve H}^{E}_{Z}$, $\boldsymbol{\rm\breve H}^{J}_{Z,q}$, and $\boldsymbol{\rm \breve H}^{N}_{Z,k}$, we can derive the EMC-Active RIS-Assisted anti-jamming communication model with DA.

 In fact, we can manipulate the tunable impedance network $\boldsymbol{\rm {Z}}_A$ in the EMC-Active RIS-assisted anti-jamming system by selecting an appropriate DA $\boldsymbol{\rm {Z}}_D$ to decouple the $M$ REs ports originally terminated by $\boldsymbol{\rm {Z}}_A$, which also implies that the reconstructed voltage-current relationship satisfies $\boldsymbol{\rm {v}}_A^D = Z_0\boldsymbol{\rm {I}}_M\boldsymbol{\rm {i}}_A^D$. The power matching network (PMN) \cite[Eq.~(103)]{Ivrlac_Nossek_2010} emerges as a DA that satisfies the aforementioned conditions, being both lossless and reciprocal. We employ the PMN to decouple the $M$ REs ports terminated by $\boldsymbol{\rm {Z}}_A$, which is designed as
 \begin{equation}\label{PMN}\small
 \boldsymbol{\rm {Z}}_D
 =\begin{bmatrix}
 \boldsymbol{\rm {0}} & -j\sqrt{Z_0}\operatorname{Re}\{\boldsymbol{\rm {Z}}_{AA}\}^{\frac{1}{2}} \\
 -j\sqrt{Z_0}\operatorname{Re}\{\boldsymbol{\rm {Z}}_{AA}\}^{\frac{1}{2}} & -j\operatorname{Im}\{\boldsymbol{\rm {Z}}_{AA}\}
 \end{bmatrix}.
 \end{equation}
 Substituting Eq.~(\ref{PMN}) into (\ref{ZAD_1}), we can get the updated tunable impedance $\boldsymbol{\rm {Z}}_A^D$ as
 
 \vspace{-3mm}
 \begin{small}
 	\begin{align}\label{ZAD_2}
 	\boldsymbol{\rm {Z}}_A^D\hspace{-1mm} =-j\operatorname{Im}\{\boldsymbol{\rm {Z}}_{AA}\}\hspace{-0.5mm}+\hspace{-0.5mm}Z_0\operatorname{Re}\{\boldsymbol{\rm {Z}}_{AA}\}^{\frac{1}{2}}\boldsymbol{\rm {Z}}_A^{-1}\operatorname{Re}\{\boldsymbol{\rm {Z}}_{AA}\}^{\frac{1}{2}}\hspace{-0.5mm},
 	\end{align}
 \end{small}

\vspace{-3mm}
 When the DA is loaded on the front end of the EMC-Active RIS, we derive the novel EMC-Active RIS-assisted anti-jamming system model with DA in the following theorem.
  \begin{theorem}\label{theorem_DA}
 	By employing the PMN and substituting Eq.~(\ref{ZAD_2}) into the equivalent channels represented by the $Z$-parameter in Theorem~\ref{theorem_Z}, we derive the redefined equivalent channel in EMC-Active RIS-assisted anti-jamming system model with DA, as follows:
 	
 	\vspace{-3mm}
 	\begin{small}
 		\begin{align}\label{DA_model1}
 		&\boldsymbol{\rm H}^{E}_{Z,k}\triangleq\frac{1}{2Z_0}\left(\boldsymbol{\rm Z}_{BU,k}-\frac{1}{2}\boldsymbol{\rm \widetilde Z}_{RU,k}\boldsymbol{\rm  \widetilde \Gamma}_{A}\boldsymbol{\rm \widetilde Z}_{BR}\right),\\
 		&\boldsymbol{\rm H}^{J}_{Z,qk}\triangleq\frac{1}{2Z_0}\left(\boldsymbol{\rm Z}_{JU,k}-\frac{1}{2}\boldsymbol{\rm \widetilde Z}_{RU,k}\boldsymbol{\rm  \widetilde \Gamma}_{A}\boldsymbol{\rm \widetilde Z}_{JR,q}\right),\label{HZJDA}\\
 		&\boldsymbol{\rm H}^{N}_{Z,k}\triangleq\frac{1}{2Z_0}\left(\boldsymbol{\rm Z}_{RU,k}-\frac{1}{2}\boldsymbol{\rm \widetilde Z}_{RU,k}\boldsymbol{\rm  \widetilde \Gamma}_{A}\boldsymbol{\rm \bar Z}_{AA}\right),\label{HZNDA}
 		\end{align}
 	\end{small}where we introduce the definitions as follows:
 
 \vspace{-3mm}
 \begin{small}
 	\begin{align}
 	&\boldsymbol{\rm \widetilde Z}_{RU,k}\triangleq\boldsymbol{\rm  Z}_{RU,k}\operatorname{Re}\{\boldsymbol{\rm {Z}}_{AA}\}^{-\frac{1}{2}},~\boldsymbol{\rm  \widetilde \Gamma}_{A}\triangleq\boldsymbol{\rm  \Gamma}_{A}+\boldsymbol{\rm  I}_{M}\notag\\
 	&\boldsymbol{\rm \widetilde Z}_{BR}\triangleq\operatorname{Re}\{\boldsymbol{\rm {Z}}_{AA}\}^{-\frac{1}{2}}\boldsymbol{\rm  Z}_{BR},~\boldsymbol{\rm \widetilde Z}_{JR,q}\triangleq\operatorname{Re}\{\boldsymbol{\rm {Z}}_{AA}\}^{-\frac{1}{2}}\boldsymbol{\rm  Z}_{JR,q},\notag\\
 	&\boldsymbol{\rm \bar Z}_{AA}\triangleq\operatorname{Re}\{\boldsymbol{\rm {Z}}_{AA}\}^{-\frac{1}{2}}\boldsymbol{\rm Z}_{AA}^+,~\boldsymbol{\rm \widetilde Z}_{AA}\triangleq\boldsymbol{\rm Z}_{AA}^+\operatorname{Re}\{\boldsymbol{\rm {Z}}_{AA}\}^{-\frac{1}{2}}.
 	\end{align}
 \end{small}

 	\vspace{-6mm}
 	In addition, for the output signal at the EMC-Active RIS, we have
 	
 	\vspace{-3mm}
 	\begin{small}
 		\begin{align}\label{H_ZADA}
 		&\boldsymbol{\rm \breve {H}}_{Z}^{E}\triangleq\frac{1}{2Z_0}\left(\boldsymbol{\rm Z}_{BR}-\frac{1}{2}\boldsymbol{\rm \widetilde Z}_{AA}\boldsymbol{\rm  \widetilde \Gamma}_{A}\boldsymbol{\rm \widetilde Z}_{BR}\right),\\
 		&\boldsymbol{\rm \breve {H}}_{Z,q}^{J}\triangleq\frac{1}{2Z_0}\left(\boldsymbol{\rm Z}_{JR,q}-\frac{1}{2}\boldsymbol{\rm \widetilde Z}_{AA}\boldsymbol{\rm  \widetilde \Gamma}_{A}\boldsymbol{\rm \widetilde Z}_{JR,q}\right),\\
 		&\boldsymbol{\rm \breve {H}}_{Z}^{N}\triangleq\frac{1}{2Z_0}\left(\boldsymbol{\rm Z}_{AA}^+-\frac{1}{2}\boldsymbol{\rm \widetilde Z}_{AA}\boldsymbol{\rm  \widetilde \Gamma}_{A}\boldsymbol{\rm \bar Z}_{AA}\right).\label{H_ZA2DA}
 		\end{align}
 	\end{small}
 \end{theorem}
 \textit{Proof:} By substituting $\boldsymbol{\rm Z}^{D}_{A}$ for $\boldsymbol{\rm Z}_{A}$ and performing algebraic transformations, we can derive the redefined equivalent channel. The specific process is omitted due to space constraints.

 \begin{remark}
 	The EMC-Active RIS-assisted anti-jamming system model with DA that integrates the advantages of $Z$-parameters and $S$-parameters, which is designed to characterize RIS-assisted anti-jaming communication links with electromagnetic MC. The core concept of this model is to construct a lossless and reciprocal DA to convert the inherent MC effects among REs into a new decoupled channel matrix. It is noteworthy that this decoupled channel model structurally aligns with systems that maintain no MC among traditional ideal RIS REs. Therefore, optimization methods developed in ideal active RIS scenarios can be directly extended to systems with MC effects. Our proposed EMC-Active RIS-assisted anti-jamming system model with DA effectively reduces the complexity introduced by traditional MC modeling, allowing MC to manifest solely as new channel mappings without compromising the portability of the algorithm framework itself. Furthermore, since the DA depends solely on the MC property of REs, it remains static and independent of the operating scenario. Consequently, the elements $\boldsymbol{\rm Z}_{D}$ of the DA are fixed and require calibration only once during array fabrication \cite{DN1}.
 \end{remark}

 For the EMC-Active RIS-assisted anti-jamming system model with DA, we can derive the SINR of the $k$th UE, denoted by $\gamma_k$ as $\gamma_k=\frac{\left\vert\boldsymbol{\rm H}^{E}_{Z,k}\boldsymbol{\rm w}_k\right\vert^2}{\sum\limits_{j\neq k,j=1}^{K}\hspace{-1mm}\left\vert\boldsymbol{\rm H}^{E}_{Z,k}\boldsymbol{\rm w}_j\right\vert^2\hspace{-2mm}+J_k+\sigma^2_k}$, where $J_k=\sum_{q=1}^{Q}\left\vert\boldsymbol{\rm H}^{J}_{Z,qk}\boldsymbol{\rm w}_{J,qk}\right\vert^2\hspace{-1.5mm}+\sigma^2_{R}\left\Vert\boldsymbol{\rm H}^{N}_{Z,k}\right\Vert^2$. Then, the achievable sum-rate can be obtained as $R = \sum_{k=1}^{K}R_k=\sum_{k=1}^{K}\mbox{log}_2(1+\gamma_{k}).$

  \subsection{Spatially Correlated Channel Model}
  \label{sec:Problem_formulation_2}
  
  In practical propagation environments, communication links typically contain both line-of-sight (LOS) and non-line-of-sight (NLOS) components. Furthermore, when nodes employ antenna structures with element spacings less than $\lambda_c/2$, spatial correlation effects inevitably arise, leading to rank-deficient channels and diminished achievable diversity gains. Therefore, to accurately characterize multipath propagation characteristics and the impact of correlation, we employ a generalized spatially correlated Rician fading model \cite{Spatially} to describe all relevant transmission links. This model has been widely adopted in RIS-assisted communication research to capture channel randomness and structural features.

  In practical deployments, due to the random positioning of UEs and the absence of RF links capable of signal processing in active RIS, BSs and active RIS struggle to obtain precise CSI related to UEs. Thus, the transmission matrix within the system can only access imperfect CSI. Furthermore, the lack of cooperation between the BS and malicious jammers prevents the acquisition of perfect CSI for channels associated with potential malicious jammers. Utilizing methods such as the array-based phase rotation scheme \cite{Sun1}, the system can typically only infer the approximate range of the jamming direction by observing jamming power, and further derive corresponding statistical characteristics. Therefore, in system-level design and optimization, we employ the statistical description of the channels (i.e., statistical CSI) as prior information to avoid the high overhead associated with precise channel estimation.
  
  Based on the generalized spatially correlated Rician fading model, the transmission terms can be given as \cite{DN1}:
  
  \vspace{-3mm}
  \begin{small}
  	\begin{align}\label{channel_1}
  	\frac{\boldsymbol{\rm {Z}}_{\varpi}}{Z_0}=\sqrt{\frac{\kappa_{\varpi}PL_{\varpi}}{\kappa_{\varpi}+1}}\boldsymbol{\rm {Z}}_{\varpi}^L+\sqrt{\frac{PL_{\varpi}}{\kappa_{\varpi}+1}}\boldsymbol{\rm {Z}}_{\varpi}^N,
  	\end{align}
  \end{small}where $\kappa_{\varpi}$, $PL_{\varpi}$, $\boldsymbol{\rm {Z}}_{\varpi}^L$, and $\boldsymbol{\rm {Z}}_{\varpi}^N$ represent Rician factor, the path loss, and the LOS component and NLOS component of $\boldsymbol{\rm {Z}}_{\varpi}$ for $\varpi\in\left\{\{BU,k\}, \{JU,qk\}, \{RU,k\}, \{BR\}, \left\{JR,q\right\}\right\}$, respectively. The LOS component $\boldsymbol{\rm {Z}}_{\varpi}^L$ is constructed as the product of the array responses between the trasmitter and receiver, which remains constant over the channel coherence time. The ULA response of BS and jammers can be written as

 \vspace{-5mm}
 \begin{small}
	\begin{align}
	\boldsymbol{\rm {a}}_{L}(\iota_{\varpi_3})=\left[ e^{j\frac{2\pi}{\lambda_c}x_{\varpi_3}\sin(\iota_{\varpi_3})},\cdots,e^{j\frac{2\pi}{\lambda_c}d_{\varpi_3}\sin(\iota_{\varpi_3})}\right]^T\hspace{-1mm},
	\end{align}
 \end{small}where $\varpi_3\in\left\{\{B\}, \{J,q\}\right\}$, and $x_{n}$ $\left(x_{\varpi_3}=\left\{N,N_J\right\}\right) $ and $\iota_{\varpi_3}$ indicate the position of the $n$th antenna and the Angle of Departure (AoD) of the BS/jammer, respectively. For EMC-Active RIS, the UPA response is defined as $\boldsymbol{\rm {a}}_{P}(\varphi_{R},\phi_{R})=\boldsymbol{\rm {a}}_{M_v}(\varphi_{R},\phi_{R})\otimes\boldsymbol{\rm {a}}_{M_h}(\varphi_{R})$, in which $\varphi_{R}$ and $\phi_{R}$ denote the the azimuth AoA/AoD and elevation AoA/AoD, respectively. Therefore, the LOS components are constructed as $\boldsymbol{\rm {Z}}_{BR}^L=\boldsymbol{\rm {a}}_{P}(\varphi_{R}^A,\phi_{R}^A)\boldsymbol{\rm {a}}_{L}^T(\iota_{B}^D)$, $\boldsymbol{\rm {Z}}_{BU,k}^L=\boldsymbol{\rm {a}}_{L}^T(\iota_{B}^D)$, $\boldsymbol{\rm {Z}}_{JR,q}^L=\boldsymbol{\rm {a}}_{P}(\varphi_{R}^A,\phi_{R}^A)\boldsymbol{\rm {a}}_{L}^T(\iota_{J,q}^D)$, $\boldsymbol{\rm {Z}}_{JU,qk}^L=\boldsymbol{\rm {a}}_{L}^T(\iota_{J,q}^D)$, and $\boldsymbol{\rm {Z}}_{RU,k}^L=\boldsymbol{\rm {a}}_{P}(\varphi_{R}^D,\phi_{R}^D)$, where the superscripts $D$ and $A$ represent the AoD and AoA,
respectively. 

%

  
  In addition, the NLOS component can be modeled as the Kronecker separable correlation model \cite{KP}, where the correlation between channel fading depends on the Kronecker product of the corresponding transmitter and receiver antenna correlations. Thus, the NLOS components are denoted by $\boldsymbol{\rm {Z}}_{BR}^N=\boldsymbol{\rm {\Sigma}}_{R}^{\frac{1}{2}}\boldsymbol{\rm \bar {Z}}_{BR}^N\boldsymbol{\rm {\Sigma}}_{B}^{\frac{1}{2}}$, $\boldsymbol{\rm {Z}}_{BU,k}^N=\boldsymbol{\rm {\Sigma}}_{B}^{\frac{1}{2}}\boldsymbol{\rm \bar {Z}}_{BU,k}^N$, $\boldsymbol{\rm {Z}}_{JR,q}^N=\boldsymbol{\rm {\Sigma}}_{R}^{\frac{1}{2}}\boldsymbol{\rm \bar {Z}}_{JR,q}^N\boldsymbol{\rm {\Sigma}}_{J,q}^{\frac{1}{2}}$, $\boldsymbol{\rm {Z}}_{JU,qk}^N=\boldsymbol{\rm {\Sigma}}_{J,q}^{\frac{1}{2}}\boldsymbol{\rm \bar {Z}}_{JU,qk}^N$, and $\boldsymbol{\rm {Z}}_{RU,k}^N=\boldsymbol{\rm {\Sigma}}_{R}^{\frac{1}{2}}\boldsymbol{\rm \bar {Z}}_{RU,k}^N$ ,where $\boldsymbol{\rm {\Sigma}}_{B}$, $\boldsymbol{\rm {\Sigma}}_{R}$, and $\boldsymbol{\rm {\Sigma}}_{J,q}$ represent the positive definite spatial correlation matrices of BS, EMC-Active RIS, and jammers, respectively. The notations $\boldsymbol{\rm \bar {Z}}_{\varpi}^N\sim\mathcal{CN}(0,1)$ denotes
  i.i.d. complex Gaussian matrices with zero mean and unit variance entries. Among numerous spatial correlation modeling approaches, the exponential model is widely adopted due to its structural simplicity. It relies solely on single parameter to generate the complete correlation matrix, thereby offering high operational feasibility. Despite its formal simplicity, extensive simulation results from prior work \cite{Channel1} demonstrate that the exponential correlation model effectively characterizes the spatial correlation properties of ULA in realistic environments. Thus, we emply the exponential model to construct the spatial correlation matrices for BS and jammers, as follows:
  
  
  \vspace{-3mm}
  \begin{small}
  	\begin{align}
  	\boldsymbol{\rm {\Sigma}}_{\varpi_3}(i_1,i_2)=
  	\left\{
  	\begin{aligned}
  	&\delta_{\varpi_3}^{\left\vert i_1-i_2\right\vert}, ~~\text{if}~i_1\ge i_2;\\
  	&\left(\delta_{\varpi_3}^*\right)^{\left\vert i_1-i_2\right\vert}, ~~\text{if}~ i_1 < i_2,
  	\end{aligned}
  	\right.
  	\end{align}
  \end{small}where $\delta_{\varpi_3}= r_{\varpi_3}e^{j\phi_{\varpi_3}}$ with $\phi_{\varpi_3}$ and $r_{\varpi_3}\in[0,1)$ being the horizontal or vertical AoA/AoD and the correlation coefficients for $\varpi_3\in\left\{\{B\}, \{J,q\}\right\}$. For EMC-Active RIS, $\boldsymbol{\rm {\Sigma}}_{R}$ is the deterministic Hermitian symmetric positive semi-definite correlation matrix \cite[Proposition 1]{Channel2}, where the $(i,j)$th entry of the correlation matrix, is expressed as $\boldsymbol{\rm {\Sigma}}_{R}(i_1,i_2)=\mbox{sinc}\left(\frac{2\left\Vert \boldsymbol{\rm u}_i-\boldsymbol{\rm u}_j\right\Vert^2}{\lambda_c}\right)$, where $\lambda_c$ is the wavelength of the incident wave, and $\mbox{sinc}(x)=\mbox{sin}(x)/(\pi x)$. Besides, the notation $\boldsymbol{\rm u}_o$, $o\in\left\{i_1,i_2\right\}$ represents the coordinates of the row index position of the $o$th RE, which is given as $ \boldsymbol{\rm u}_o=\left[0, \mbox{mod}(o-1,M_h)d_h, \left\lfloor \frac{(o-1)}{M_h}d_v \right\rfloor\right]^T$, where $d_h$ is the horizontal width and $d_v$ is the vertical height of REs.

  For simplicity of description, we provide a simplified representation \cite{OAM1} of the channel in Eq.~(\ref{channel_1}): it can be modeled as a matrix variable following a complex Gaussian distribution with mean $\boldsymbol{\rm \mu}_{\varpi}=\varepsilon_\varpi^L\boldsymbol{\rm {Z}}_{\varpi}^L$ and variance $\left(\left(\varepsilon_\varpi^N\right)^2\boldsymbol{\rm {\Sigma}}_{\varpi_1}\otimes\boldsymbol{\rm {\Sigma}}_{\varpi_2}\right)$, where $\varepsilon_\varpi^L=\sqrt{\frac{\kappa_{\varpi}PL_{\varpi}}{\kappa_{\varpi}+1}}$ and $\varepsilon_\varpi^N=\sqrt{\frac{PL_{\varpi}}{\kappa_{\varpi}+1}}$. The notation $\otimes$ denotes the Kronecker product. Thus, we have $\boldsymbol{\rm {Z}}_{\varpi}\sim\mathcal{CN}\left(\boldsymbol{\rm \mu}_{\varpi},~\left(\varepsilon_\varpi^N\right)^2\boldsymbol{\rm {\Sigma}}_{\varpi_1}\otimes\boldsymbol{\rm {\Sigma}}_{\varpi_2}\right).$

 \subsection{Problem Formulation}
 \label{sec:Problem_formulation_1}
 
 For the EMC-Active RIS-assisted anti-jamming system with DA, taking into account the situation that BS and EMC-Active RIS can only access the statistical CSI, the ergodic achievable rate formulated as $R_e=\sum_{k=1}^{K}\mathbb{E} \left[R_k\right]=\sum_{k=1}^{K}\mathbb{E} \left[\log_2\left(1+\gamma_{k}\right)\right]$. Since the expected term in the ergodic achievable rate involves integration over the SINR, it is typically difficult to obtain an analytical closed-form solution. Therefore, we utilize Jensen's inequality and leverage the statistical properties of the correlation matrix to bound the achievable rate from above, thereby deriving its upper-bound expression $\widetilde{R}_e$ as $\widetilde{R}_e=\sum_{k=1}^{K}\log_2\left(1+\mathbb{E}\left[\gamma_{k}\right]\right)$.
 
  In addition, both the power of the active RIS output signal $\boldsymbol{\rm y}_A$ and the amplification magnitude of the signal are constrained by the active load at the EMC-Active RIS. Therefore, when the statistical CSI is known, the maximum amplification power constraint and the amplitude constraint for single RE are expressed as follows:
  
  \vspace{-3mm}
  \begin{small}
  	\begin{align}\label{constraint_1}
  	&P_A=\mathbb{E}\left[\sum_{k=1}^{K}\left\Vert\boldsymbol{\rm \breve {H}}_{Z}^{E}\boldsymbol{\rm w}_k\right\Vert^2+\sum_{k=1}^{K}\sum_{q=1}^Q\left\Vert\boldsymbol{\rm \breve H}^{J}_{Z,q}\boldsymbol{\rm w}_{J,q}\right\Vert^2\right.\notag\\
  	&\qquad\qquad\qquad\qquad\qquad\left.+\sigma^2_{R}\left\Vert\boldsymbol{\rm \breve H}^{N}_{Z}\right\Vert^2\right]\le P_{A,\max},\\
  	&\Gamma_{A,m}\le\Gamma_{A,\max}, \forall m\in \mathcal{M},\label{constraint_2}
  	\end{align}
  \end{small}where $P_{A,\max}$ and $\Gamma_{A,\max}$ are the amplification power budget and the maximum amplitude, respectively. Furthermore, the phase modulation in actual RIS hardware is typically implemented by adjustable components with limited precision, whose achievable phase states are constrained by the number of quantization bits and circuit complexity, making continuous phase variation unrealistic. Therefore, compared to the idealized continuous-phase assumption, the discrete-phase model more accurately reflects the physical implementation conditions and engineering constraints of RIS, making the designed optimization algorithms more practically feasible and valuable for real-world applications \cite{PS}. Also, with respect to discrete phase shifts, $\theta_{m}$ is selected from a set of finite discrete values uniformly sampled from the interval $\left(0,2\pi\right]$, which is given as follows:
  \begin{equation}\label{constraint_3}\small
  \theta_{m}\in\mathcal{F}=\left\{0,\frac{2\pi}{2^b},2\frac{2\pi}{2^b},\cdots,\left(2^b-1\right)\frac{2\pi}{2^b}\right\},
  \end{equation}
  where $b$ presents the number of bits used to indicate the number of phase shift levels $2^b$.
  
 
  In this subsection, we focus on the joint design of transmit beamforming and the reflection coefficient matrix of the EMC-active RIS, aiming to maximize the upper-bound of the ergodic achievable rate. As such, the optimization problem is described as follows:
  
  \vspace{-3mm}
 \begin{subequations}\label{Problem:P1}\small
	\begin{alignat}{2}
	\textbf{\textit{P}1:}
	&\mathop{\max}\limits_{\boldsymbol{\rm w}_k,\boldsymbol{\rm \Gamma}_A} \widetilde{R}_e\left(\boldsymbol{\rm w}_k,\boldsymbol{\rm \Gamma}_A\right)  \notag \\
	{\rm{s.t.}}:&1).\ \sum_{k=1}^{K}\Vert\boldsymbol{\rm {w}}_{k}\Vert^2\le P_{\max};\\
	&2).\ (\ref{constraint_1}),~(\ref{constraint_2}),~\text{and}~(\ref{constraint_3}).
	\end{alignat}
 \end{subequations}
 \begin{remark}
 	By comparing the original coupled optimization problem \cite[Eq.~(33)]{Cao3} with $\textbf{\textit{P}1}$, it can be seen that the negative impact of MC in RE on the solvability of the original coupled optimization problem is eliminated in $\textbf{\textit{P}1}$. This demonstrates that DA fundamentally transforms the coupled EMC-Active RIS system into an tractable decoupled representation that preserves the essential EM behavior while eliminating the optimization bottleneck caused by MC, thereby enabling scalable optimization.
 \end{remark}

 \section{Proposed Alternating Optimization}
 \label{sec:Alternating_optimization_perfect_CSI}
 
 In this section, in order to simplify the optimization problem $\textbf{\textit{P}1}$, we first derive its expression after eliminating expectations, then employ an alternating iteration (AO) algorithm to decouple optimization variables for efficient solution. To solve optimization problem $\textbf{\textit{P}1}$, we first need to simplify the expectation operations within the objective function and constraints in the following proposition.
 
 \begin{proposition}\label{proposition}
 	The objective function for $\textbf{\textit{P}1}$ can be simplified into a form without expectation operations, as follows:
 	
 	\vspace{-5mm}
 	\begin{small}
 		\begin{align}\label{Re}
 		\widetilde{R}_e=\sum_{k=1}^{K}\log_2\left(1+\frac{\boldsymbol{\rm w}_k^H\boldsymbol{\rm C}_{1,k}\boldsymbol{\rm w}_k}{\sum\limits_{j\neq k,j=1}^{K}\hspace{-1mm}\boldsymbol{\rm w}_j^H\boldsymbol{\rm C}_{1,k}\boldsymbol{\rm w}_j+{\bar J}_k+\sigma^2_k}\right),
 		\end{align}
 	\end{small}where ${\bar J}_k=\sum_{q=1}^{Q}\boldsymbol{\rm w}_{J,q}^H\boldsymbol{\rm C}_{2,qk}\boldsymbol{\rm w}_{J,q}+\sigma^2_{R}\mbox{Tr}\left(\boldsymbol{\rm C}_{3,k}\right)$ and $\mbox{Tr}(\cdot)$ is the trace of the matrix. To simplify notations of the objective function for $\textbf{\textit{P}1}$, we introduce the following notations
 
   \vspace{-3mm}
   \begin{small}
 	 \begin{align}
 	 \boldsymbol{\rm C}_{1,k}\triangleq&\boldsymbol{\rm R}_{BU,k}+\frac{1}{\left(4Z_0\right)^2}\boldsymbol{\rm\widetilde \mu}_{BR}^H\boldsymbol{\rm \widetilde\Gamma}_{A}^H\boldsymbol{\rm\widetilde R}_{RU,k}\boldsymbol{\rm \widetilde\Gamma}_{A}\boldsymbol{\rm \widetilde \mu}_{BR}\notag\\
 	 &+\frac{1}{\left(4Z_0\right)^2}\left(\varepsilon_{BR}^N\right)^2\mbox{Tr}\left(\boldsymbol{\rm \widetilde\Gamma}_{A}^H\boldsymbol{\rm\widetilde R}_{RU,k}\boldsymbol{\rm \widetilde\Gamma}_{A}\boldsymbol{\rm\widetilde {\Sigma}}_{R}\right)\boldsymbol{\rm {\Sigma}}_{B},\notag\\
 	 \boldsymbol{\rm C}_{2,qk}\triangleq&\boldsymbol{\rm R}_{JU,qk}+\frac{1}{\left(4Z_0\right)^2}\boldsymbol{\rm\widetilde \mu}_{JR,q}^H\boldsymbol{\rm \widetilde\Gamma}_{A}^H\boldsymbol{\rm\widetilde R}_{RU,k}\boldsymbol{\rm \widetilde\Gamma}_{A}\boldsymbol{\rm \widetilde \mu}_{JR,q}\notag\\
 	 &+\frac{1}{\left(4Z_0\right)^2}\left(\varepsilon_{JR,q}^N\right)^2\mbox{Tr}\left(\boldsymbol{\rm \widetilde\Gamma}_{A}^H\boldsymbol{\rm\widetilde R}_{RU,k}\boldsymbol{\rm \widetilde\Gamma}_{A}\boldsymbol{\rm\widetilde {\Sigma}}_{R}\right)\boldsymbol{\rm {\Sigma}}_{J,q},\notag\\
 	 \boldsymbol{\rm C}_{3,k}\triangleq&\frac{\boldsymbol{\rm R}_{RU,k}}{\left(2Z_0\right)^2}+\frac{1}{\left(4Z_0\right)^2}\boldsymbol{\rm\bar Z}_{AA}^H\boldsymbol{\rm \widetilde\Gamma}_{A}^H\boldsymbol{\rm\widetilde R}_{RU,k}\boldsymbol{\rm \widetilde\Gamma}_{A}\boldsymbol{\rm\bar Z}_{AA}^H,
 	 \end{align}
  \end{small}where we have

  \vspace{-3mm}
  \begin{small}
  	\begin{align}
  	&\boldsymbol{\rm\widetilde \mu}_{\varpi}=\boldsymbol{\rm \mu}_{\varpi}\operatorname{Re}\{\boldsymbol{\rm {Z}}_{AA}\}^{-\frac{1}{2}},\boldsymbol{\rm\widetilde \Sigma}_{R}=\operatorname{Re}\{\boldsymbol{\rm {Z}}_{AA}\}^{-\frac{1}{2}}\boldsymbol{\rm \Sigma}_{R}\operatorname{Re}\{\boldsymbol{\rm {Z}}_{AA}\}^{-\frac{1}{2}}\notag\\
  	&\boldsymbol{\rm\widetilde R}_{\varpi}=\operatorname{Re}\{\boldsymbol{\rm {Z}}_{AA}\}^{-\frac{1}{2}} \left(\boldsymbol{\rm \mu}_{\varpi}^H\boldsymbol{\rm \mu}_{\varpi}+\left(\varepsilon_{\varpi}^N\right)^2\boldsymbol{\rm {\Sigma}}_{\varpi_3}\right)\operatorname{Re}\{\boldsymbol{\rm {Z}}_{AA}\}^{-\frac{1}{2}},\notag\\
  	&\boldsymbol{\rm R}_{\varpi}=\frac{1}{\left(2Z_0\right)^2}\left(\boldsymbol{\rm \mu}_{\varpi}^H\boldsymbol{\rm \mu}_{\varpi}+\left(\varepsilon_{\varpi}^N\right)^2\boldsymbol{\rm {\Sigma}}_{\varpi_3}\right),
  	\end{align}
  \end{small}for $\varpi\in\left\{\{BU,k\}, \{JU,qk\}, \{RU,k\}\right\}$ and $\varpi_3\in\left\{\{B\}, \{J,q\}\right\}$. For the constraint (\ref{constraint_1}), we can also simplify it as follows:

  \vspace{-5mm}
  \begin{small}
	\begin{align}\label{constraint_1new}
	P_A=\sum_{k=1}^{K}\boldsymbol{\rm w}_k^H\boldsymbol{\rm C}_{4}\boldsymbol{\rm w}_k&+\sum_{k=1}^{K}\sum_{q=1}^Q\boldsymbol{\rm w}_{J,q}^H\boldsymbol{\rm C}_{5,q}\boldsymbol{\rm w}_{J,q}\notag\\
	&+\sigma^2_{R}\mbox{Tr}\left(\boldsymbol{\rm C}_{6}\right)\le P_{A,\max},
	\end{align}
  \end{small}in which we also introduce definitions

  \vspace{-3mm}
  \begin{small}
	\begin{align}
	\boldsymbol{\rm C}_{4}\triangleq&\boldsymbol{\rm R}_{BR}+\frac{1}{\left(4Z_0\right)^2}\boldsymbol{\rm\widetilde \mu}_{BR}^H\boldsymbol{\rm \widetilde\Gamma}_{A}^H\boldsymbol{\rm \widetilde Z}_{AA}^H\boldsymbol{\rm \widetilde Z}_{AA}\boldsymbol{\rm \widetilde\Gamma}_{A}\boldsymbol{\rm \widetilde \mu}_{BR}\notag\\
	&+\frac{1}{\left(4Z_0\right)^2}\left(\varepsilon_{BR}^N\right)^2\mbox{Tr}\left(\boldsymbol{\rm \widetilde\Gamma}_{A}^H\boldsymbol{\rm \widetilde Z}_{AA}^H\boldsymbol{\rm \widetilde Z}_{AA}\boldsymbol{\rm \widetilde\Gamma}_{A}\boldsymbol{\rm\widetilde {\Sigma}}_{R}\right)\boldsymbol{\rm {\Sigma}}_{B},\notag\\
	\boldsymbol{\rm C}_{5,q}\triangleq&\boldsymbol{\rm R}_{JR,q}+\frac{1}{\left(4Z_0\right)^2}\boldsymbol{\rm\widetilde \mu}_{JR,q}^H\boldsymbol{\rm \widetilde\Gamma}_{A}^H\boldsymbol{\rm \widetilde Z}_{AA}^H\boldsymbol{\rm \widetilde Z}_{AA}\boldsymbol{\rm \widetilde\Gamma}_{A}\boldsymbol{\rm \widetilde \mu}_{JR,q}\notag\\
	&+\frac{1}{\left(4Z_0\right)^2}\left(\varepsilon_{JR,q}^N\right)^2\mbox{Tr}\left(\boldsymbol{\rm \widetilde\Gamma}_{A}^H\boldsymbol{\rm \widetilde Z}_{AA}^H\boldsymbol{\rm \widetilde Z}_{AA}\boldsymbol{\rm \widetilde\Gamma}_{A}\boldsymbol{\rm\widetilde {\Sigma}}_{R}\right)\boldsymbol{\rm {\Sigma}}_{J,q},\notag\\
	\boldsymbol{\rm C}_{6}\triangleq&\frac{\boldsymbol{\rm R}_{AA}^+}{\left(2Z_0\right)^2}+\frac{1}{\left(4Z_0\right)^2}\boldsymbol{\rm\bar Z}_{AA}^H\boldsymbol{\rm \widetilde\Gamma}_{A}^H\boldsymbol{\rm \widetilde Z}_{AA}^H\boldsymbol{\rm \widetilde Z}_{AA}\boldsymbol{\rm \widetilde\Gamma}_{A}\boldsymbol{\rm\bar Z}_{AA}^H,
	\end{align}
  \end{small}where $\boldsymbol{\rm R}_{\varpi}=\frac{1}{\left(2Z_0\right)^2}\left(\boldsymbol{\rm \mu}_{\varpi}^H\boldsymbol{\rm \mu}_{\varpi}+\left(\varepsilon_{\varpi}^N\right)^2M\boldsymbol{\rm {\Sigma}}_{\varpi_3}\right)$ for $\varpi\in\left\{\{BR\}, \{JR,q\}\right\}$, and $\boldsymbol{\rm R}_{AA}^+=(\boldsymbol{\rm  Z}_{AA}^+)^H\boldsymbol{\rm Z}_{AA}^+$.
 \end{proposition}
 \textit{Proof:} Please see Appendix B.$\hfill\blacksquare$ 
 
 After eliminating expectations based on Proposition~{\ref{proposition}}, $\textbf{\textit{P}1}$ can be rewritten as follows:
 \begin{subequations}\label{Problem:P1-A}\small
 	\begin{alignat}{2}
 	\textbf{\textit{P}1-A:}
 	&\mathop{\max}\limits_{\boldsymbol{\rm w}_k,\boldsymbol{\rm \Gamma}_A} \widetilde{R}_e\left(\boldsymbol{\rm w}_k,\boldsymbol{\rm \Gamma}_A\right)  \notag \\
 	{\rm{s.t.}}:&\ (\ref{Problem:P1}\text{a}),~ (\ref{constraint_1new}),~(\ref{constraint_2}),~\text{and}~(\ref{constraint_3}),
 	\end{alignat}
 \end{subequations}
 where the expression of $\widetilde{R}_e$ is given by Eq.~(\ref{Re}). To further handle the fractions in the objective function, we first reconstruct $\boldsymbol{\rm C}_{1,k}=\boldsymbol{\rm H}_{C,k}^H\boldsymbol{\rm H}_{C,k}$,  where $\boldsymbol{\rm H}_{C,k}=\left[\boldsymbol{\rm H}_{C1,k}; \boldsymbol{\rm H}_{C2,k}; \boldsymbol{\rm H}_{C3,k}\right]$, with $\boldsymbol{\rm H}_{C1,k}$, $\boldsymbol{\rm H}_{C2,k}$, and $\boldsymbol{\rm H}_{C3,k}$ being $\boldsymbol{\rm H}_{C1,k}=\boldsymbol{\rm R}_{BU,k}^{\frac{1}{2}}\in\mathbb{C}^{N\times N}$, $\boldsymbol{\rm H}_{C2,k}=\frac{1}{4Z_0}\boldsymbol{\rm\widetilde R}_{RU,k}^{\frac{1}{2}}\boldsymbol{\rm\widetilde \Gamma}_{A}\boldsymbol{\rm\widetilde \mu}_{BR}\in\mathbb{C}^{M\times N}$, and $\boldsymbol{\rm H}_{C3,k}=\frac{\varepsilon_{BR}^N}{4Z_0}\left\Vert\boldsymbol{\rm\bar \Gamma}_{A}\boldsymbol{\rm\widetilde\Sigma}_{RU,k}^{\frac{1}{2}}\right\Vert_2\boldsymbol{\rm\widetilde\Sigma}_{B}^{\frac{1}{2}}\in\mathbb{C}^{N\times N}$, respectively, where $\boldsymbol{\rm\bar \Gamma}_{A}$ is the column vector composed of the main diagonal elements of $\boldsymbol{\rm\widetilde \Gamma}_{A}$ and $\boldsymbol{\rm\widetilde\Sigma}_{RU,k}=\boldsymbol{\rm\widetilde R}_{RU,k}\odot\boldsymbol{\rm\widetilde\Sigma}_{R}^T$. Note that $\boldsymbol{\rm H}_{C1,k}$ and $\boldsymbol{\rm H}_{C2,k}$ hold because $\boldsymbol{\rm R}_{BU,k}$ and $\boldsymbol{\rm\widetilde R}_{RU,k}$ are positive
 definite Hermitian matrices with an arithmetic square root; $\boldsymbol{\rm H}_{C3,k}$ stands for our invoking property: $\mbox{Tr}\left(\boldsymbol{\rm \widetilde\Gamma}_{A}^H\boldsymbol{\rm\widetilde R}_{RU,k}\boldsymbol{\rm \widetilde\Gamma}_{A}\boldsymbol{\rm\widetilde {\Sigma}}_{R}\right)=\boldsymbol{\rm\bar \Gamma}_{A}^T\boldsymbol{\rm\widetilde\Sigma}_{RU,k}\boldsymbol{\rm\bar \Gamma}_{A}$, where $\boldsymbol{\rm\widetilde\Sigma}_{RU,k}$ is also a positive definite Hermitian matrix based on schur product theory. Thus, we have
 
 \vspace{-3mm}
 \begin{small}
 	\begin{align}\label{change1}
 	\boldsymbol{\rm w}_k^H\boldsymbol{\rm C}_{1,k}\boldsymbol{\rm w}_k=\left\Vert\boldsymbol{\rm H}_{C,k}\boldsymbol{\rm w}_k\right\Vert^2.
 	\end{align}
 \end{small}
 \vspace{-5mm}

 Substituting Eq.~(\ref{change1}) into (\ref{Re}), we can then utilize multidimensional Lagrangian Dual Transform (mLDT) and quadratic transform (QT) \cite{danbi} to handle the logarithmic and fractional forms in the objective function of $\textbf{\textit{P}1-A}$. Specifically, the objective function of $\textbf{\textit{P}1-A}$ can be transformed by invoking auxiliary variables $\boldsymbol{\rm {\omega}} =\left[\omega_{1},\cdots,\omega_{K}\right]^T$ and $\boldsymbol{\rm {\nu}}=\left[\boldsymbol{\rm {\nu}}_{1},\cdots,\boldsymbol{\rm {\nu}}_{K}\right]^T$ with $\boldsymbol{\rm {\nu}}_{k}\in\mathbb{C}^{{2N+M}\times 1}$ and ignoring irrelevant items as follows:
 
 \vspace{-3mm}
 \begin{small}
 	\begin{align}
 	f_{OF}&=\sum_{k=1}^{K}\left[\mbox{ln}(1+\omega_{k})+2\sqrt{1+\omega_{k}}\mbox{Re}\left\{\boldsymbol{\rm {\nu}}_{k}^H\boldsymbol{\rm {H}}_{C,k}\boldsymbol{\rm {w}}_{k}\right\}\right.\nonumber\\
 	&-\omega_{k}-\left.\left(\boldsymbol{\rm {\nu}}_{k}^H\boldsymbol{\rm {\nu}}_{k}\right)\hspace{-1mm}\left(\hspace{-1mm}\textstyle\sum\limits_{j=1}^{K}\left\Vert\boldsymbol{\rm {H}}_{C,k}\boldsymbol{\rm {w}}_{j}\right\Vert^2+\bar{J}_{k}+\sigma^2_{k}\right)\right]\hspace{-1mm}.
 	\end{align}
 \end{small}Thus, $\textbf{\textit{P}1-A}$ is reconstructed as follows:
 \begin{subequations}\label{Problem:P1-B}\small
 	\begin{alignat}{2}
 	\textbf{\textit{P}1-B:}
 	&\mathop{\max}\limits_{\boldsymbol{\rm w}_k,\boldsymbol{\rm \Gamma}_A,\boldsymbol{\rm {\omega}},\boldsymbol{\rm {\nu}}} f_{OF}\left(\boldsymbol{\rm {w}}_{k},\boldsymbol{\rm \widetilde{\Gamma}}_{A}, \boldsymbol{\rm {\omega}},\boldsymbol{\rm {\nu}}\right)  \notag \\
 	{\rm{s.t.}}:&\ (\ref{Problem:P1}\text{a}),~ (\ref{constraint_1new}),~(\ref{constraint_2}),~\text{and}~(\ref{constraint_3}),
 	\end{alignat}
 \end{subequations}
 Note that by setting $\partial f_{OF}/\partial \omega_{k}=0$ and $\partial f_{OF}/\partial\boldsymbol{\rm {\nu}}_{k}=0$, the optimal $\omega_{k}^{\star}$ and $\boldsymbol{\rm {\nu}}_{k}^{\star}$ are obtained, respectively.
 
 \vspace{-3mm}
 \subsection{Optimize $\boldsymbol{\rm {w}}_k$}
 \label{subsec:W}
 
 For simplicity, we define $\boldsymbol{\rm w}=[\boldsymbol{\rm w}_1^T,\cdots,\boldsymbol{\rm w}_K^T]^T$. When the variable $\boldsymbol{\rm \widetilde{\Gamma}}_A$ is fixed, the subproblem with respect to $\boldsymbol{\rm {w}}$ is reconstructed as follows:\vspace{-2mm}
 \begin{subequations}\label{P2}\small
 	\begin{alignat}{2}
 	\textbf{\textit{P}2:}
 	&\mathop{\max}\limits_{\boldsymbol{\rm {w}}} \ \mbox{Re}\left\{\boldsymbol{\rm {\zeta}}^H\boldsymbol{\rm {w}}\right\}-\boldsymbol{\rm {w}}^H\boldsymbol{\rm {\Upsilon}}_1\boldsymbol{\rm {w}} \notag \\
 	{\rm{s.t.}}:&1).\ \boldsymbol{\rm {w}}^H\boldsymbol{\rm {w}}\le {P}_{\max};\\
 	&2).\ \boldsymbol{\rm {w}}^H\boldsymbol{\rm {\Upsilon}}_2\boldsymbol{\rm {w}}\le \bar{P}_{A,\max},
 	\end{alignat}
 \end{subequations}
  where we define
  
  \vspace{-4mm}
  \begin{small}
  	\begin{align}
  	&\boldsymbol{\rm {\zeta}}\triangleq[\boldsymbol{\rm {\zeta}}_1^T,\cdots,\boldsymbol{\rm {\zeta}}_K^T]^T,\boldsymbol{\rm {\zeta}}_k=2\sqrt{1+\gamma_k}\boldsymbol{\rm {\nu}}_k^H\boldsymbol{\rm {H}}_{C,k},\notag\\
  	&\boldsymbol{\rm {\Upsilon}}_1\triangleq\boldsymbol{\rm {I}}_k\otimes\left(\sum_{k=1}^{K}\boldsymbol{\rm {\nu}}_k^H\boldsymbol{\rm {\nu}}_k\boldsymbol{\rm {H}}_{C,k}^H\boldsymbol{\rm {H}}_{C,k}\right),~\boldsymbol{\rm {\Upsilon}}_2\triangleq\boldsymbol{\rm {I}}_k\otimes\boldsymbol{\rm{C}}_{4}\notag\\
  	&\bar{P}_{A,\max} \triangleq{P}_{A,\max}\hspace{-1mm}-\hspace{-1mm}\sum_{k=1}^{K}\sum_{q=1}^Q\hspace{-1mm}\boldsymbol{\rm w}_{J,q}^H\boldsymbol{\rm C}_{5,q}\boldsymbol{\rm w}_{J,q}\hspace{-1mm}-\hspace{-1mm}\sigma^2_{R}\mbox{Tr}\left(\boldsymbol{\rm C}_{6}\right).
  	\end{align}
  \end{small}Owing to its standard quadratically constrained quadratic programming (QCQP) structure, $\textbf{\textit{P}2}$ can in principle be solved using the CVX toolkit. Nevertheless, the computational burden associated with CVX-based solutions is unbearable for practical implementation. To this end, we instead develop a low-complexity approach to obtain an approximate optimal solution by leveraging the ADMM framework \cite{ADMM} via the following lemma.
  \begin{lemma}\label{lemma0}
  	Consider an optimization problem:
  	
  	\vspace{-3mm}
  	\begin{small}
  		\begin{align}
  		&\mathop{\min}\limits_{\boldsymbol{\rm {x}}} \ \boldsymbol{\rm {x}}^H\boldsymbol{\rm {K}}_1\boldsymbol{\rm {x}}-2\mbox{Re}\left\{\boldsymbol{\rm {k}}^H\boldsymbol{\rm {x}}\right\}~{\rm{s.t.}}:\boldsymbol{\rm {x}}^H\boldsymbol{\rm {K}}_2\boldsymbol{\rm {x}}\le P,
  		\end{align}
  	\end{small}where $\boldsymbol{\rm {K}}_1$ and $\boldsymbol{\rm {K}}_2$ denote positive definite Hermitian matrices. Let $\boldsymbol{\rm \widetilde{K}}_1=\boldsymbol{\rm {K}}_2^{-\frac{1}{2}}\boldsymbol{\rm {K}}_1\boldsymbol{\rm {K}}_2^{-\frac{1}{2}}$ with eigenvalue decomposition $\boldsymbol{\rm \widetilde{K}}_1=\boldsymbol{\rm \widetilde{V}}^H\boldsymbol{\rm \widetilde{\Xi}}\boldsymbol{\rm \widetilde{V}}$ and $\boldsymbol{\rm \widetilde{k}}=\boldsymbol{\rm U}^H\boldsymbol{\rm{k}}$ with $\boldsymbol{\rm U}=\boldsymbol{\rm {K}}_2^{-\frac{1}{2}}\boldsymbol{\rm \widetilde{V}}^H$. The optimal value of $\boldsymbol{\rm {x}}^\star$ can be obtained by examining the following two cases: 1). If $\left(\boldsymbol{\rm {K}}_1^{-1}\boldsymbol{\rm {k}}\right)^H\boldsymbol{\rm {K}}_2\boldsymbol{\rm {K}}_1^{-1}\boldsymbol{\rm {k}}\le P$, $\lambda_x^{\star}=0$, $\boldsymbol{\rm {x}}^\star=\boldsymbol{\rm {K}}_1^{-1}\boldsymbol{\rm {k}}$; 2). Otherwise, $\lambda_x^{\star}>0$, $\boldsymbol{\rm {x}}^\star=\left(\boldsymbol{\rm {K}}_1^{-1}+\lambda_x^{\star}\boldsymbol{\rm {K}}_2\right)\boldsymbol{\rm {k}}$. The optimal value of $\lambda_x^{\star}$ can be determined using the bisection method or Newton's method from $\sum_{i}\frac{\left\vert\widehat{k}_i\right\vert^2}{\Xi_i+\lambda_x} = P$, where $\widehat{k}_i$ is the element of $\boldsymbol{\rm \widetilde{k}}$ and $\Xi_i$ is the eigenvalue of $\boldsymbol{\rm \widetilde{K}}_1$.
  \end{lemma}
  \textit{Proof:} The proof of this conclusion is analogous to the proof of the KKT conditions in \cite[Lemma 2]{KKT}, with the modification being the application of the constant transformation for generalized eigenvalue decomposition to $\boldsymbol{\rm \widetilde{K}}_1$. $\hfill\blacksquare$

  To leverage Lemma~\ref{lemma0} for developing an efficient solution, we adopt the ADMM framework by introducing a copy $\boldsymbol{\rm w}_C$ of $\boldsymbol{\rm w}$ (i.e., $\boldsymbol{\rm w}_C = \boldsymbol{\rm w}$) to address the constraints. The augmented Lagrangian (AL) function of $\textbf{\textit{P}2}$ is constructed as
 
 \vspace{-5mm}
  \begin{subequations}\label{P2-A}\small
  	\begin{alignat}{2}
  	\textbf{\textit{P}2-A:}
  	&\mathop{\min}\limits_{\boldsymbol{\rm {w}},\boldsymbol{\rm {w}}_C,\boldsymbol{\rm {\lambda}}_w} \ \boldsymbol{\rm {w}}^H\boldsymbol{\rm {\Upsilon}}_1\boldsymbol{\rm {w}}-\mbox{Re}\left\{\boldsymbol{\rm {\zeta}}^H\boldsymbol{\rm {w}}\right\}+\frac{\rho}{2}\left\Vert\boldsymbol{\rm {w}}-\boldsymbol{\rm {w}}_C\right\Vert^2\notag\\
  	&\qquad\qquad\qquad\qquad\qquad+\mbox{Re}\left\{\boldsymbol{\rm {\lambda}}_w^H\left(\boldsymbol{\rm {w}}-\boldsymbol{\rm {w}}_C\right)\right\} \notag \\
  	{\rm{s.t.}}:&1).\ \boldsymbol{\rm {w}}^H\boldsymbol{\rm {w}}\le {P}_{\max};\\
  	&2).\ \boldsymbol{\rm {w}}^H_C\boldsymbol{\rm {\Upsilon}}_2\boldsymbol{\rm {w}}_C\le \bar{P}_{A,\max},
  	\end{alignat}
  \end{subequations}
 where $\boldsymbol{\rm {\lambda}}_w$ and $\rho$ represent the AL multiplier and a predefined positive constant. Then, we can alternately update variables $\boldsymbol{\rm {w}}$, $\boldsymbol{\rm {w}}_C$, and $\boldsymbol{\rm {\lambda}}_w$. The sub-optimization problem with regard to $\boldsymbol{\rm {w}}$ can be written as follow:
 
 \vspace{-3mm}
 \begin{small}
 	\begin{align}\label{w1}
 	&\mathop{\min}\limits_{\boldsymbol{\rm {w}}} \ \boldsymbol{\rm {w}}^H\boldsymbol{\rm \widetilde{\Upsilon}}_1\boldsymbol{\rm {w}}-2\mbox{Re}\left\{\boldsymbol{\rm {\zeta}}_W^H\boldsymbol{\rm {w}}\right\}~{\rm{s.t.}}:(\ref{P2-A}{\text{a}}),
 	\end{align}
 \end{small}where $\boldsymbol{\rm {\zeta}}_W=\boldsymbol{\rm {\zeta}}+\frac{\rho\boldsymbol{\rm {w}}_C-\boldsymbol{\rm {\lambda}}_w}{2}$ and $\boldsymbol{\rm \widetilde{\Upsilon}}_1=\boldsymbol{\rm {\Upsilon}}_1+\frac{\rho}{2}\boldsymbol{\rm {I}}$. Therefore, we can apply Lemma~\ref{lemma0} to deal with sub-optimization problem with regard to $\boldsymbol{\rm {w}}$.
  
  For the sub-optimization problem for $\boldsymbol{\rm {w}}_C$, we have 
  
  \vspace{-3mm}
  \begin{small}
  	\begin{align}\label{w2}
  	&\mathop{\min}\limits_{\boldsymbol{\rm {w}}_C} \left\Vert\boldsymbol{\rm {w}}_C\right\Vert^2\hspace{-1mm}-2\mbox{Re}\left\{\boldsymbol{\rm {\zeta}}_C^H\boldsymbol{\rm {w}}_C\right\}{\rm{s.t.}}:(\ref{P2-A}{\text{b}}),
  	\end{align}
  \end{small}where $\boldsymbol{\rm {\zeta}}_C=\boldsymbol{\rm {w}}+\rho^{-1}\boldsymbol{\rm {\lambda}}_w$. The sub-optimization problem for $\boldsymbol{\rm {w}}_C$ can also be solved by using Lemma~\ref{lemma0}. In addition, $\boldsymbol{\rm {\lambda}}_w$ can be updated via $\boldsymbol{\rm {\lambda}}_w^{i+1}=\boldsymbol{\rm {\lambda}}_w^i+\rho\left(\boldsymbol{\rm {w}}-\boldsymbol{\rm {w}}_C\right)$. Note that the proof of the optimality of the solution for the subproblem with respect to $\boldsymbol{\rm {w}}$ is provided in \cite[Theorem 1]{KKT}.
  
  \vspace{-2mm}

  \subsection{Optimize $\boldsymbol{\rm {\Gamma}}_A$}
  \label{subsec:A}
 
  With the variable $\boldsymbol{\rm {w}}$ is fixed, the subproblem with regard to $\boldsymbol{\rm {\Gamma}}_A$, which is recast as follows:
  
  \vspace{-3mm}
  \begin{small}
  	\begin{align}\label{Problem:P3}
  	\textbf{\textit{P}3:}
  	\mathop{\max}\limits_{\boldsymbol{\rm \Gamma}_A} f_{OF}\left(\boldsymbol{\rm \widetilde{\Gamma}}_{A}\right)  
  	{\rm{s.t.}}:\ (\ref{constraint_1new}),~(\ref{constraint_2}),~\text{and}~(\ref{constraint_3}),
  	\end{align}
  \end{small}The primary challenge in solving $\textbf{\textit{P}3}$ lies in the non-convexity of the constraints imposed by the amplitude and phase shift matrices in $\boldsymbol{\rm \Gamma}_A$. Therefore, based on the construction of $\boldsymbol{\rm \Gamma}_A$ in Eq.~(\ref{PA1}), we perform alternating iterative optimization of the amplitude matrix at RA and the phase shift matrix at PS.

 \subsubsection{Optimize the amplitude matrix $\boldsymbol{\rm \Lambda}$}
 
 For ease of notation, we introduce the new definitions: $\boldsymbol{\rm {\nu}}_k\triangleq[\boldsymbol{\rm {\nu}}_{k,1};\boldsymbol{\rm {\nu}}_{k,2};\boldsymbol{\rm {\nu}}_{k,3}]$, $\boldsymbol{\rm {g}}_{\Lambda,k}\triangleq\boldsymbol{\rm {\mu}}_{BR}\boldsymbol{\rm {w}}_{k}$, $\boldsymbol{\rm {\psi}}_j\triangleq\frac{1}{4Z_0}\varepsilon_{BR}^N\boldsymbol{\rm {\Sigma}}_{B}^{\frac{1}{2}}\boldsymbol{\rm {w}}_{j}$, $\boldsymbol{\rm {B}}_{\Lambda,kj}=\frac{1}{4Z_0}\boldsymbol{\rm\widetilde R}_{RU,k}^{\frac{1}{2}}\mbox{Diag}(\boldsymbol{\rm {g}}_{\Lambda,j})$, and $\ell_k(\boldsymbol{\rm \bar{\Lambda}})=\left\Vert\boldsymbol{\rm {\Sigma}}_{\Lambda,k}^{\frac{1}{2}}\boldsymbol{\rm \bar{\Lambda}}\right\Vert_2$, $\forall k,j\in\mathcal{K}$, where $\boldsymbol{\rm\Sigma}_{\Lambda,k}=(\boldsymbol{\rm\widetilde \Theta}^H\boldsymbol{\rm\widetilde R}_{RU,k}\boldsymbol{\rm\widetilde \Theta})\odot\boldsymbol{\rm\widetilde\Sigma}_{R}^T$ and $\boldsymbol{\rm\widetilde \Theta}=L_{P\hspace{-0.3mm}S}^{2}\exp\left(j2\boldsymbol{\rm \Theta}\right)$. The notation $\boldsymbol{\rm\bar \Lambda}$ is the column vector composed of the main diagonal elements of $(\boldsymbol{\rm\Lambda}+\boldsymbol{\rm I}_M)$. The terms $\mbox{Re}\left\{\boldsymbol{\rm {\nu}}_{k}^H\boldsymbol{\rm {H}}_{C,k}\boldsymbol{\rm {w}}_{k}\right\}$ and $\left\Vert\boldsymbol{\rm {H}}_{C,k}\boldsymbol{\rm {w}}_{j}\right\Vert^2$ of $f_{OF}$ can be redescribed as
 
 \vspace{-3mm}
 \begin{small}
 	\begin{align}\label{gongshi1}
 	\mbox{Re}\left\{\boldsymbol{\rm {\nu}}_{k}^H\boldsymbol{\rm {H}}_{C,k}\boldsymbol{\rm {w}}_{k}\right\}=&\mbox{Re}\left\{\boldsymbol{\rm {\nu}}_{k,1}^H
 	\boldsymbol{\rm R}_{BU,k}^{\frac{1}{2}}\boldsymbol{\rm {w}}_{k}+\boldsymbol{\rm {\nu}}_{k,2}^H\boldsymbol{\rm {B}}_{\Lambda,k}\boldsymbol{\rm\bar \Lambda}\right.\notag\\
 	&\left.+\boldsymbol{\rm {\nu}}_{k,3}^H\ell_k(\boldsymbol{\rm \bar{\Lambda}})\boldsymbol{\rm {\psi}}_k\right\}, \forall k\in\mathcal{K},\notag\\
 	\left\Vert\boldsymbol{\rm {H}}_{C,k}\boldsymbol{\rm {w}}_{j}\right\Vert^2=&\left\Vert\boldsymbol{\rm R}_{BU,k}^{\frac{1}{2}}\boldsymbol{\rm {w}}_{j}\right\Vert^2+\left\Vert\boldsymbol{\rm {B}}_{\Lambda,kj}\boldsymbol{\rm\bar \Lambda}\right\Vert^2\notag\\
 	&+\boldsymbol{\rm\bar \Lambda}^H\boldsymbol{\rm {\Sigma}}_{\Lambda,k}\boldsymbol{\rm\bar \Lambda}\left\Vert\boldsymbol{\rm {\psi}}_j\right\Vert^2, \forall k,j\in\mathcal{K}.
 	\end{align}
 \end{small}Due to the non-convexity of $\ell_k(\boldsymbol{\rm \bar{\Lambda}})$, $\mbox{Re}\left\{\boldsymbol{\rm {\nu}}_{k}^H\boldsymbol{\rm {H}}_{C,k}\boldsymbol{\rm {w}}_{k}\right\}$ remains a non-convex expression. Therefore, we consider employing the SCA method to efficiently handle $\ell_k(\boldsymbol{\rm \bar{\Lambda}})$. By applying a Taylor series expansion, we approximate it using its lower-bound, as follows:

\vspace{-3mm}
\begin{small}
	\begin{align}\label{gongshi1_2}
	\ell_k(\boldsymbol{\rm \bar{\Lambda}})\ge\ell_{\Lambda,k}^{(\tau)}+\mbox{Re}\left\{\left(\triangledown\ell_{\Lambda,k}^{(\tau)}\right)^H\left(\boldsymbol{\rm \bar{\Lambda}}-\boldsymbol{\rm \bar{\Lambda}}^{(\tau)}\right)\right\},
	\end{align}
\end{small}where $\ell_{\Lambda,k}^{(\tau)}=\ell_k(\boldsymbol{\rm \bar{\Lambda}}^{(\tau)})$, $\triangledown\ell_{\Lambda,k}^{(\tau)}=\boldsymbol{\rm {\Sigma}}_{\Lambda,k}^{\frac{1}{2}}\boldsymbol{\rm \bar{\Lambda}}^{(\tau)}/\ell_{\Lambda,k}^{(\tau)}$ at $\tau$th iteration. By introducing similar definitions
 $\boldsymbol{\rm {e}}_{\Lambda,q}\triangleq\boldsymbol{\rm {\mu}}_{JR,q}\boldsymbol{\rm {w}}_{J,q}$, $\boldsymbol{\rm {\chi}}_{q}\triangleq\frac{1}{4Z_0}\varepsilon_{JR,q}^N\boldsymbol{\rm {\Sigma}}_{J,q}^{\frac{1}{2}}\boldsymbol{\rm {w}}_{J,q}$, $\boldsymbol{\rm\bar {B}}_{\Lambda,qk}=\frac{1}{4Z_0}\boldsymbol{\rm\widetilde R}_{RU,k}^{\frac{1}{2}}\mbox{Diag}(\boldsymbol{\rm {e}}_{\Lambda,q})$, $\boldsymbol{\rm\widetilde\Sigma}_{\Lambda,k}=(\frac{1}{\left(4Z_0\right)^2}\boldsymbol{\rm\widetilde \Theta}^H\boldsymbol{\rm\widetilde R}_{RU,k}\boldsymbol{\rm\widetilde \Theta})\odot(\boldsymbol{\rm\bar Z}_{AA}\boldsymbol{\rm\bar Z}_{AA}^H)^T$, $\boldsymbol{\rm\bar\Sigma}_{\Lambda,A}=(\boldsymbol{\rm\widetilde \Theta}^H\boldsymbol{\rm\widetilde Z}_{AA}\boldsymbol{\rm\widetilde Z}_{AA}^H\boldsymbol{\rm\widetilde \Theta})\odot\boldsymbol{\rm\widetilde\Sigma}_{R}^T$, $\boldsymbol{\rm\widehat\Sigma}_{\Lambda,A}=(\frac{1}{\left(4Z_0\right)^2}\boldsymbol{\rm\widetilde \Theta}^H\boldsymbol{\rm\widetilde Z}_{AA}\boldsymbol{\rm\widetilde Z}_{AA}^H\boldsymbol{\rm\widetilde \Theta})\odot(\boldsymbol{\rm\bar Z}_{AA}\boldsymbol{\rm\bar Z}_{AA}^H)^T$, $\boldsymbol{\rm\widetilde {B}}_{\Lambda,k}=\frac{1}{4Z_0}\boldsymbol{\rm\widetilde Z}_{AA}\mbox{Diag}(\boldsymbol{\rm {g}}_{\Lambda,k})$, and $\boldsymbol{\rm\widehat {B}}_{\Lambda,q}=\frac{1}{4Z_0}\boldsymbol{\rm\widetilde Z}_{AA}\mbox{Diag}(\boldsymbol{\rm {e}}_{\Lambda,q})$, the term $\bar J_k$ of objecive function and the constraint (\ref{constraint_1new}) are restructured in a similar manner due to sharing the same structure as the $\left\Vert\boldsymbol{\rm {H}}_{C,k}\boldsymbol{\rm {w}}_{j}\right\Vert^2$ term. Due to space constraints, we omit the specific derivation here.
 
 By substituting Eq.~(\ref{gongshi1}), (\ref{gongshi1_2}), and reconstructed versions of $\bar J_k$ and constraint (\ref{constraint_1new}) and dropping the irrelevant terms, the subproblem with regard to $\boldsymbol{\rm {\Lambda}}$ is written as follows:
 
 \vspace{-3mm}
 \begin{subequations}\label{P3-A}\small
 	\begin{alignat}{2}
 	\textbf{\textit{P}3-A:}
 	&\mathop{\min}\limits_{\boldsymbol{\rm\bar {\Lambda}}} \ \boldsymbol{\rm\bar {\Lambda}}^H\boldsymbol{\rm {Y}}_{\Lambda}\boldsymbol{\rm\bar {\Lambda}}-\mbox{Re}\left\{\boldsymbol{\rm {t}}^H_{\Lambda}\boldsymbol{\rm\bar {\Lambda}}\right\}\notag \\
 	{\rm{s.t.}}:&1).\ \bar{\Lambda}_m \le \Gamma_{A,\max}+1;\\
 	&2).\ \boldsymbol{\rm\bar {\Lambda}}^H\boldsymbol{\rm\widetilde {Y}}_{\Lambda}\boldsymbol{\rm\bar {\Lambda}}\le \widetilde{P}_{A,\max},
 	\end{alignat}
 \end{subequations}
 where the parameters are defined as follows:\begin{small}
 	\begin{align}
 	\boldsymbol{\rm {t}}_{\Lambda}\triangleq&\sum_{k=1}^{K}2\sqrt{1+\omega_{k}}\left(\boldsymbol{\rm {B}}_{\Lambda,k}^H\boldsymbol{\rm {\nu}}_{k,2}+\mbox{Re}\left\{\boldsymbol{\rm {\nu}}_{k,3}^H\boldsymbol{\rm {\psi}}_k\triangledown\ell_{\Lambda,k}^{(\tau)}\right\}\right),\nonumber\\
 	\boldsymbol{\rm {Y}}_{\Lambda}\triangleq&\sum_{k=1}^{K}\left(\boldsymbol{\rm {\nu}}_{k}^H\boldsymbol{\rm {\nu}}_{k}\right)\left(\sum_{j=1}^{K}\left(\boldsymbol{\rm {B}}_{\Lambda,kj}^H\boldsymbol{\rm {B}}_{\Lambda,kj}+\boldsymbol{\rm {\Sigma}}_{\Lambda,k}\left\Vert\boldsymbol{\rm {\psi}}_j\right\Vert^2\right)\right.\nonumber\\
 	&+\hspace{-1mm}\left.\sum_{q=1}^{Q}\left(\boldsymbol{\rm\bar {B}}_{\Lambda,qk}^H\boldsymbol{\rm \bar{B}}_{\Lambda,qk}+\boldsymbol{\rm {\Sigma}}_{\Lambda,k}\left\Vert\boldsymbol{\rm {\chi}}_q\right\Vert^2\right)+\sigma^2_R\boldsymbol{\rm\widetilde\Sigma}_{\Lambda,k}\right)\hspace{-1mm},\nonumber\\
 	\boldsymbol{\rm\widetilde {Y}}_{\Lambda}\triangleq&\sum_{k=1}^{K}\left(\boldsymbol{\rm\widetilde {B}}_{\Lambda,k}^H\boldsymbol{\rm\widetilde {B}}_{\Lambda,k}+\boldsymbol{\rm\bar\Sigma}_{\Lambda,A}\left\Vert\boldsymbol{\rm {\psi}}_k\right\Vert^2\right)\nonumber\\
 	&+\sum_{k=1}^{K}\sum_{q=1}^{Q}\left(\boldsymbol{\rm\widehat {B}}_{\Lambda,q}^H\boldsymbol{\rm\widehat {B}}_{\Lambda,q}+\boldsymbol{\rm\bar\Sigma}_{\Lambda,A}\left\Vert\boldsymbol{\rm {\chi}}_q\right\Vert^2\right)+\sigma^2_R\boldsymbol{\rm\widehat\Sigma}_{\Lambda,A},\nonumber\\
 	\widetilde{P}_{A,\max}&\triangleq{P}_{A,\max}-\sum_{k=1}^{K}\sum_{q=1}^{Q}\boldsymbol{\rm {w}}_{J,q}\boldsymbol{\rm {R}}_{JR,q}\boldsymbol{\rm {w}}_{J,q}\nonumber\\
 	&-\sum_{k=1}^{K}\boldsymbol{\rm {w}}_{k}\boldsymbol{\rm {R}}_{BR}\boldsymbol{\rm {w}}_{k}-\sigma^2_R\mbox{Diag}\left(\frac{1}{(2Z_0)^2}\boldsymbol{\rm R}_{AA}^+\right).
 	\end{align}
 \end{small}

 To leverage Lemma~\ref{lemma0} once more for developing an efficient solution, we introduce a copy $\boldsymbol{\rm\bar {\Lambda}}_C$ of $\boldsymbol{\rm\bar {\Lambda}}$ and employ the ADMM method to address the constraint $\boldsymbol{\rm\bar {\Lambda}}=\boldsymbol{\rm\bar {\Lambda}}_C$ as follows:
 
 \vspace{-3mm}
 \begin{subequations}\label{P3-B}\small
 	\begin{alignat}{2}
 	\textbf{\textit{P}3-B:}
 	&\mathop{\min}\limits_{\boldsymbol{\rm\bar {\Lambda}},\boldsymbol{\rm\bar {\Lambda}}_C,\boldsymbol{\rm {\lambda}}_\Lambda} \ \boldsymbol{\rm\bar {\Lambda}}^H\boldsymbol{\rm {Y}}_{\Lambda}\boldsymbol{\rm\bar {\Lambda}}-\mbox{Re}\left\{\boldsymbol{\rm {t}}^H_{\Lambda}\boldsymbol{\rm\bar {\Lambda}}\right\}+\frac{\rho_{\Lambda}}{2}\left\Vert\boldsymbol{\rm\bar {\Lambda}}-\boldsymbol{\rm\bar {\Lambda}}_C\right\Vert^2\notag\\
 	&\qquad\qquad\qquad\qquad\qquad+\mbox{Re}\left\{\boldsymbol{\rm {\lambda}}_{\Lambda}^H\left(\boldsymbol{\rm\bar {\Lambda}}-\boldsymbol{\rm\bar {\Lambda}}_C\right)\right\} \notag \\
 	{\rm{s.t.}}:&1).\ \bar{\Lambda}_{C,m} \le \Gamma_{A,\max}+1;\\
 	&2).\ \boldsymbol{\rm\bar {\Lambda}}^H\boldsymbol{\rm\widetilde {Y}}_{\Lambda}\boldsymbol{\rm\bar {\Lambda}}\le \widetilde{P}_{A,\max},
 	\end{alignat}
 \end{subequations}
 Then, we can alternately update variables $\boldsymbol{\rm\bar {\Lambda}}$, $\boldsymbol{\rm \bar{\Lambda}}_C$, and $\boldsymbol{\rm {\lambda}}_{\Lambda}$. The sub-optimization problem with regard to $\boldsymbol{\rm {\Lambda}}$ can be written as follow:
 
 \vspace{-3mm}
 \begin{small}
 	\begin{align}\label{L1}
 	&\mathop{\min}\limits_{\boldsymbol{\rm \bar{\Lambda}}} \ \boldsymbol{\rm \bar{\Lambda}}^H\boldsymbol{\rm \widetilde{Y}}_{\Lambda}\boldsymbol{\rm\bar {\Lambda}}-2\mbox{Re}\left\{\boldsymbol{\rm\widetilde {t}}_{\Lambda}^H\boldsymbol{\rm {w}}\right\}~{\rm{s.t.}}:(\ref{P3-B}{\text{b}}),
 	\end{align}
 \end{small}where $\boldsymbol{\rm \widetilde{t}}_{\Lambda}=\frac{\boldsymbol{\rm {t}}_{\Lambda}}{2}+\frac{\rho_{\Lambda}-\boldsymbol{\rm {\lambda}}_{\Lambda}}{2}$ and $\boldsymbol{\rm \widetilde{Y}}_{\Lambda}=\boldsymbol{\rm {Y}}_{\Lambda}+\frac{\rho}{2}\boldsymbol{\rm {I}}$. Therefore, we can apply Lemma~\ref{lemma0} to deal with sub-optimization problem with regard to $\boldsymbol{\rm\bar {\Lambda}}$. For the sub-optimization problem for $\boldsymbol{\rm\bar {\Lambda}}_C$, we have 

 \vspace{-3mm}
\begin{small}
	\begin{align}\label{L2}
	&\mathop{\min}\limits_{\boldsymbol{\rm \bar{\Lambda}}_C} \left\Vert\boldsymbol{\rm\bar {\Lambda}}_C\right\Vert^2\hspace{-1mm}-2\mbox{Re}\left\{\boldsymbol{\rm\bar {t}}_{\Lambda}^H\boldsymbol{\rm\bar {\Lambda}}_C\right\}{\rm{s.t.}}:(\ref{P3-B}{\text{a}}),
	\end{align}
\end{small}where $\boldsymbol{\rm \bar{t}}_{\Lambda}=\boldsymbol{\rm {t}}_{\Lambda}+\rho_{\Lambda}^{-1}\boldsymbol{\rm {\lambda}}_{\Lambda}$. The closed-form solution for $\boldsymbol{\rm\bar {\Lambda}}_C$ can be obtained as $\boldsymbol{\rm \bar{\Lambda}}_C^{\star}=\min\left\{\boldsymbol{\rm \bar{t}}_{\Lambda},\Gamma_{A,\max}+1\right\}$. In addition, $\boldsymbol{\rm {\lambda}}_{\Lambda}$ can be updated via $\boldsymbol{\rm {\lambda}}_{\Lambda}^{i+1}=\boldsymbol{\rm {\lambda}}_{\Lambda}^i+\rho_{\Lambda}\left(\boldsymbol{\rm\bar {\Lambda}}-\boldsymbol{\rm\bar {\Lambda}}_C\right)$.

 \subsubsection{Optimize the phase shift matrix $\boldsymbol{\rm \Theta}$}
 
 With the fixed $\boldsymbol{\rm \Lambda}$, the subproblem for $\boldsymbol{\rm \Theta}$ can be rewritten as
 
 \vspace{-3mm}
 \begin{small}
 	\begin{align}\label{Problem:P4}
 	\textbf{\textit{P}4:}
 	\mathop{\max}\limits_{\boldsymbol{\rm \Phi}} f_{OF}\left(\boldsymbol{\rm \Phi}\right)  
 	~{\rm{s.t.}}:\ (\ref{constraint_3}),
 	\end{align}
 \end{small}where $\boldsymbol{\rm \Phi}=\exp\left(j2\boldsymbol{\rm \Theta}\right)$ with $\boldsymbol{\rm \Phi}=[\phi_1,\cdots,\phi_M]^T$. In fact, the simplified derivation of the subproblem concerning $\boldsymbol{\rm \Phi}$ is analogous to that of the subproblem with respect to $\boldsymbol{\rm \Lambda}$. We can render it easier to solve by introducing a series of new definitions, as follows:

\vspace{-3mm}
\begin{small}
	\begin{align}\label{Problem:P4-A}
	\textbf{\textit{P}4-A:}
	\mathop{\min}\limits_{\boldsymbol{\rm \Phi}} \ \boldsymbol{\rm \Phi}^H\boldsymbol{\rm {Y}}_{\Phi}\boldsymbol{\rm \Phi}-\mbox{Re}\left\{\boldsymbol{\rm {t}}^H_{\Phi}\boldsymbol{\rm \Phi}\right\}  
	~{\rm{s.t.}}:\ (\ref{constraint_3}),
	\end{align}
\end{small}The parameters introduced here are consistent with those in the optimization problem concerning $\boldsymbol{\rm \Lambda}$, where the replacement of the subscript from $\Lambda$ to $\Phi$ signifies that only $\boldsymbol{\rm\widetilde \Theta})$ is substituted with $\boldsymbol{\rm\widetilde\Lambda}=L_{P\hspace{-0.3mm}S}^{2}(\boldsymbol{\rm\Lambda}+\boldsymbol{\rm I}_M)$ in the definitions, with the remainder unchanged. To reduce redundancy, we omit their specific expressions.

%
 
 Given the constraint that the phase shift distribution of the RIS is discrete, the continuous optimization method from Lemma~\ref{lemma0} may not yield optimal results for discrete phase scenarios. Therefore, we consider employing an AO algorithm to optimize the phase shift for each RE. Consequently, when we fix the other ($M-1$) REs of the EMC-Active RIS, the $\textbf{\textit{P}4-A}$ can be reformulated as follows:

 \vspace{-3mm}
 \begin{small}
 	\begin{align}\label{Problem:P4-B}
 	\textbf{\textit{P}4-B:}
 	\mathop{\max}\limits_{ \phi_m} \ \mbox{Re}\left\{\bar{t}_{\Phi,m}\phi_m\right\} +c_{\Phi,m}
 	~{\rm{s.t.}}:\ (\ref{constraint_3}),
 	\end{align}
 \end{small}where $\bar{t}_{\Phi,m}={t}_{\Phi,m}-2\sum_{j\neq m}Y_{\Phi,mj}\phi_j$ and $c_{\Phi,m}=\mbox{Re}\left\{\sum_{i\neq m}{t}^*_{\Phi,i}\phi_i\right\}-\sum_{i,j\neq m}\phi_i^*Y_{\Phi,ij}\phi_i-Y_{\Phi,mm}$ with ${t}_{\Phi,m}$ and $Y_{\Phi,ij}$ being the $m$th element of $\boldsymbol{\rm {t}}_{\Phi}$ and the $(i,j)$th element of $\boldsymbol{\rm {Y}}_{\Phi}$, respectively. For the $m$th RE at EMC-Active RIS, the optimal $\theta_{m}^{\star}$ can be obtained as $\theta_{m}^{\star}=\arg\min_{\theta_{m}\in\mathcal{F}}\left\vert\hspace{-3mm}\mod\hspace{-2mm}\left\{2\theta_{m}+\arg(t_{\Phi,m})+\pi,2\pi\right\}-\pi\right\vert\hspace{-1mm}.$
	Similarly, we can repeat this process for other REs and alternate iterations until convergence. It is worth noting that the proof of convergence for this AO algorithm can be found in \cite[Eq. (26)]{PS}.

 \begin{algorithm}[t]\small
 	\caption{DA-based AO algorithm for Problem \textbf{\textit{P}1-A}.} \label{alg:Framwork_all}
 	\begin{algorithmic}[1]
 		\STATE Initialize $\boldsymbol{\rm {w}}^{(0)}$,$\boldsymbol{\rm \Lambda}^{(0)}$,$\boldsymbol{\rm {\Theta}}^{(0)}$, the accuracy $\eta_{\varsigma}=10^{-3}$, the number of iterations $i=0$, and the maximum number of iterations $i_{\max}$;
 		\REPEAT
 		\STATE Compute the objective function $f_{OF} (i)$ of $\textbf{\textit{P}1-A}$;
 		\STATE Update $\boldsymbol{\rm {\omega}}^{(i+1)}$ and $\boldsymbol{\rm {\nu}}^{(i+1)}$;
 		\STATE Set $\boldsymbol{\rm {w}}^{(i+1,0)}=\boldsymbol{\rm {w}}_C^{(i+1,0)}=\boldsymbol{\rm {w}}^{(i)}$, initialize $\boldsymbol{\rm {\lambda}}_w^{(0)}$ and $i_1 = 0$;
 		\REPEAT
 		\STATE Update $\boldsymbol{\rm {w}}^{(i+1,i_1+1)}$ and $\boldsymbol{\rm {w}}_C^{(i+1,i_1+1)}$ by using Lemma~\ref{lemma0};
 		\STATE Update $\boldsymbol{\rm {\lambda}}_w^{(i_1+1)}$ and $i_1=i_1+1$;
 		\UNTIL Convergence;
 		\STATE Set $\boldsymbol{\rm {w}}^{(i+1)}=\boldsymbol{\rm {w}}^{(i+1,\infty)}$ and $\boldsymbol{\rm\bar {\Lambda}}^{(i+1,0)}=\boldsymbol{\rm\bar {\Lambda}}_C^{(i+1,0)}=\boldsymbol{\rm\bar {\Lambda}}^{(i)}$, initialize $\boldsymbol{\rm {\lambda}}_{\Lambda}^{(0)}$ and the number of iterations $i_2= 0$;
 		\REPEAT
 		\STATE Update $\boldsymbol{\rm\bar {\Lambda}}^{(i+1,i_2+1)}$ and $\boldsymbol{\rm\bar {\Lambda}}_C^{(i+1,i_2+1)}$ by solving (\ref{L1}), (\ref{L2});
 		\STATE Update $\boldsymbol{\rm {\lambda}}_{\Lambda}^{(i_2+1)}$ and $i_2=i_2+1$;
 		\UNTIL Convergence;
 		\STATE Set $\boldsymbol{\rm\bar {\Lambda}}^{(i+1)}=\boldsymbol{\rm\bar {\Lambda}}^{(i+1,\infty)}$, $\boldsymbol{\rm {\Phi}}^{(i+1,0)}=\boldsymbol{\rm {\Phi}}^{(i)}$, and $i_3= 0$;
 		\REPEAT
 		\FOR{$m=1$ to $M$}
 		\STATE Update ${\theta}_m^{(i+1,i_3+1)}$ by computing $\textbf{\textit{P}4-B}$;
 		\ENDFOR
 		\STATE $i_3=i_3+1$;
 		\UNTIL Convergence;
 		\STATE Compute the objective function $f_{OF} (i+1)$ of $\textbf{\textit{P}1-A}$; 
 		\STATE Update $i=i+1$;
 		\UNTIL $\left\vert f_{OF} (i+1)-f_{OF} (i)\right\vert<\eta_{\varsigma}$ or $i = i_{\max}$.
 	\end{algorithmic}
 \end{algorithm}

 \subsection{Complexity and Convergence Analysis}
\label{subsec:AO}
 
 The complete AO algorithm based on decoupling architecture (abbreviated as DA-based AO algorithm) for addressing problem $\textbf{\textit{P}1-A}$ is outlined in Algorithm~\ref{alg:Framwork_all}. By applying Proposition~\ref{proposition}, $\textbf{\textit{P}1}$ is reformulated as $\textbf{\textit{P}1-A}$, which is subsequently addressed through an AO strategy. The resulting DA-based AO framework operates as an iterative procedure that typically converges to a locally optimal solution. Within this framework, $\textbf{\textit{P}1-A}$ is partitioned into three tractable subproblems, namely $\textbf{\textit{P}2}$, $\textbf{\textit{P}3-A}$, $\textbf{\textit{P}4-A}$, which are solved sequentially at each iteration. It is worth noting that the proofs of the convergence of each subproblem's algorithm and the optimality of its solution are provided in their respective subsections.
 
 Next, we analyze the computational complexity of Algorithm~\ref{alg:Framwork_all}. At each iteration, the dominant computational cost arises from the optimization steps associated with $\boldsymbol{\rm {w}}$, $\boldsymbol{\rm \Lambda}$ and $\boldsymbol{\rm {\Theta}}$. The complexity of optimizing $\boldsymbol{\rm w}$ and $\boldsymbol{\rm \Lambda}$ are governed by the ADMM method, which are given as $\mathcal{O}_1=\mathcal{O}\left(i_{1,\max}\left(KN\right)^{3}\right)$ and $\mathcal{O}_2=\mathcal{O}\left(i_{2,\max}\left(M\right)^{3}\right)$, respectively \cite{ADMM}. In addition, the complexity of optimizing $\boldsymbol{\rm \Theta}$ is given as $\mathcal{O}_3=\mathcal{O}\left(i_{3,\max}M\right)$. The notation $i_{s,\max}, s\in\left\{1,2,3\right\}$ represent the maximum number of iteration for subproblems. Therefore, the complexity of the overall Algorithm~\ref{alg:Framwork_all} can be obtained as $\mathcal{O}\left(i_{\max}\left(\mathcal{O}_1+\mathcal{O}_2+\mathcal{O}_3\right)\right)$.
 
 In each outer iteration, the difference in complexity between the DA-based AO algorithm and the SMaN-based AO algorithm used in the original coupled system \cite{Cao3} stems primarily from the optimization of $\boldsymbol{\rm \Lambda}$ and $\boldsymbol{\rm \Theta}$. For DA-based AO, we define the per-iteration complexity at the RIS as $C_{DA} = \mathcal{O}_2+\mathcal{O}_3$. Since $i_{2,\max}$ and $i_{3,\max}$ are typically small iteration constants and only weakly dependent on $M$, the dominant complexity order of the decoupled scheme is $\mathcal{O}(M^3)$ when $M$ becomes large. In constrast, the complexity of updating $\boldsymbol{\rm \Lambda}$ and $\boldsymbol{\rm \Theta}$ in the SMaN-based AO are $\mathcal{O}(M^4+M^2K)$ and $\mathcal{O}(M^3)$, respectively. Thus, the per-iteration complexity at the RIS is defined as $C_{SMaN}=\mathcal{O}(M^4+M^2K+M^3)$. This indicates that DA directly alters the computational structure of each iteration: the $\boldsymbol{\rm \Lambda}$ subproblem reduces the complexity from quadratic to cubic, while the $\boldsymbol{\rm \Theta}$ subproblem further reduces it from cubic to linear, thereby significantly reducing the optimization overhead per iteration.
 
 From the perspective of scaling behavior, the dominant per-iteration complexity of the original coupled scheme scales as $\mathcal{O}(M^4)$. Hence, the per-iteration complexity ratio between the two schemes satisfies $\frac{C_{SMaN}}{C_{DA}}\approx\mathcal{O}(M)$, when $K\ll M^2$. This implies that the computational advantage of the proposed DA-based AO method becomes increasingly pronounced as the RIS size grows, thereby demonstrating superior scalability for large-scale RIS systems. Furthermore, since DA simplifies the optimization structure of subproblems, the DA-based AO algorithm significantly accelerates convergence, reducing the number of outer iterations by up to three orders of magnitude compared to traditional methods. This conclusion is validated by the simulations in Section \ref{sec:simulation}.

\section{Simulation Results and Analysis}
\label{sec:simulation}

 Numerical simulations are conducted to evaluate the performance of EMC-Active RIS-assisted anti-jamming system at $2.4$ GHz. Unless otherwise specified, all system parameters are fixed throughout the simulations. The considered setup consists of a BS equipped with $N=4$ antennas,  $Q=3$ jammer equipped with $N_J = 4$, and $K=4$ single-antenna UEs. The parameters of active RIS are given the RE spacing $d_{s}=\lambda_c/4$ and $M=M_h\times M_v =100$, where $M_h=10$ and $M_v=10$. The power parameters are defined as $P_{\max}=30$ dBm, $\Gamma_{\max}^2=30$ dB, $P_{J,\max}=10$ dBm, and $\sigma_R^2=\sigma_k^2=-105$ dBm, respectively. A 3D Cartesian coordinate system is adopted to describe the network geometry, where the BS and the center of the active RIS are located at (40m,0,1m) and (0,60m,2m), respectively. UEs and jammers are  positioned along a circular trajectory with radius $r_c=15$ m centered at (20m,120m,1m) and a rectangular region restricted by the boundaries with (10m, 120m, 0) and (40m, 150m, 0), respectively. The path loss in dB is defined as $PL_{\varpi} = \beta_0 +10\beta_{\varpi} \mbox{log}_{10}(d_{\bar\chi}/d_0)$, where $\beta_0$, $\beta_{\varpi}$ and $d_{\varpi}$ are the path loss at the reference distance $d_0=1$ m, the path loss exponent, and the distance between corresponding devices. According to the settings, we can calculate the distance among BS, RIS, jammers and UEs. In addition, the channel parameters are given as ${\beta}_{BU} = {\beta}_{JU}=2.75$, ${\beta}_{BR} ={\beta}_{JR} = 2.5$, ${\beta}_{RU} = 2.2$, and $\kappa_{\varpi} = 3$ \cite{Cao2}. For the $Z$-parameter matrix containing MC and mismatch effects, ie., $\boldsymbol{\rm {Z}}_{AA}$, its expression can be given as \cite{DN1}: $\boldsymbol{\rm {Z}}_{AA}(i,j)=-\frac{Z_0e^{-j2\pi d_s\vert i-j\vert/\lambda_c}}{2\pi d_s\vert i-j\vert/\lambda_c}$, if $i\neq j$; $\boldsymbol{\rm {Z}}_{AA}(i,i)=Z_0$. All simulation curves are derived from the average of 300 random tests.

 \begin{figure}[t]
 	\vspace{-15pt}
 	\centering
 	\includegraphics[scale=0.38]{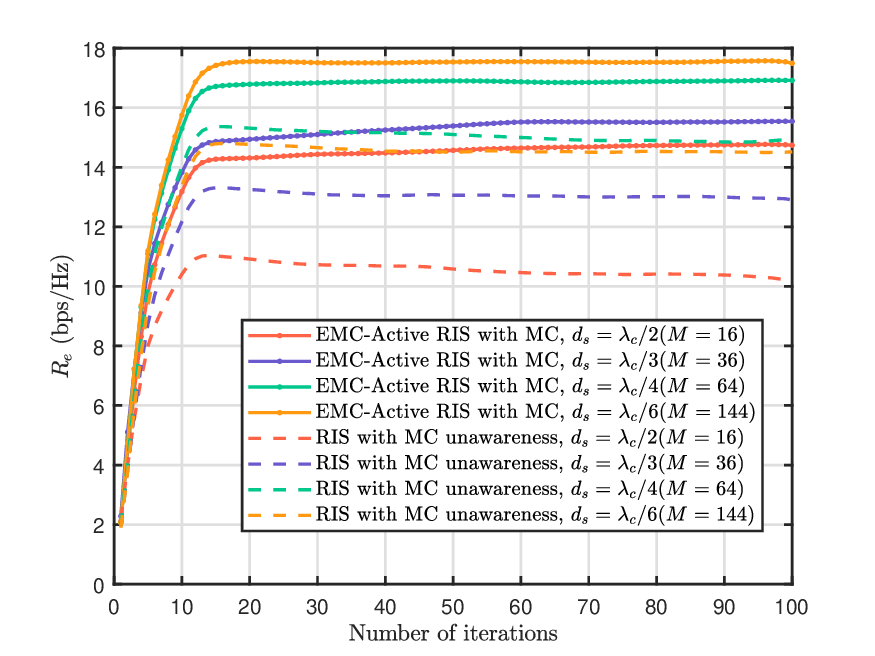}
 	\vspace{-10pt}
 	\caption{Convergence with fixed EMC-Active RIS size.}
 	\vspace{-13pt}
 	\label{fig:ite_EE}
 \end{figure}

 For comparison, the proposed EMC-Active RIS with MC is benchmarked against three baseline schemes: 1) the ideal active/passive RIS model that neglects MC \cite{Active_RIS1}; 2) the active RIS model with MC unawareness, where the parameters optimized under the ideal active RIS assumption are directly applied to EMC-Active RIS with MC; and 3) the passive RIS model with MC that explicitly accounts for MC \cite{S_parameter}. For fairness, the transmit power budget of the passive RIS-based system is set equal to the sum of the power budgets allocated to the active RIS and the BS in the active RIS-assisted system.

 \begin{figure}[t]
 	\centering
 	\includegraphics[scale=0.42]{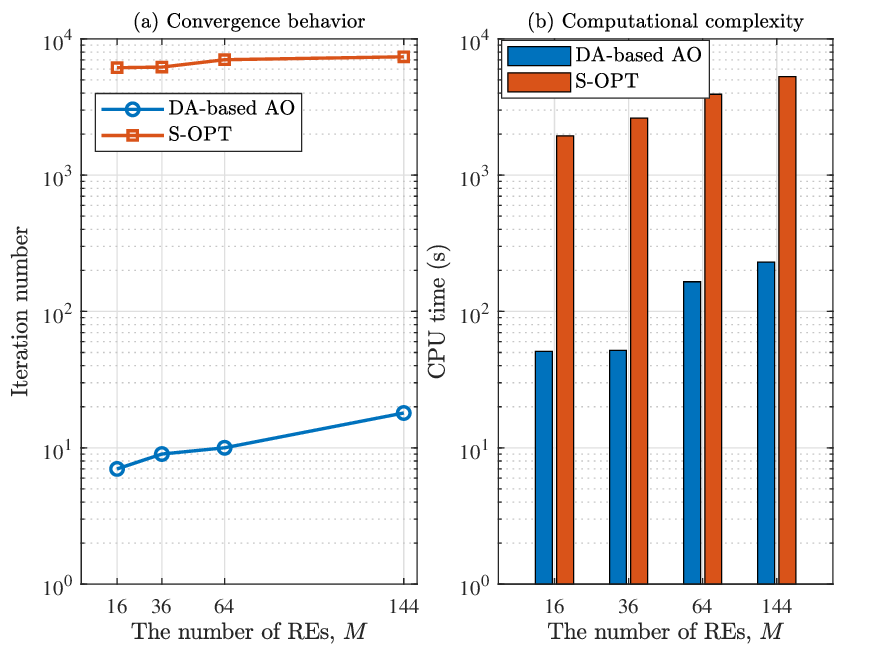}
 	\vspace{-10pt}
 	\caption{Algorithm complexity versus the number $M$ of REs.}
 	\vspace{-13pt}
 	\label{fig:ite_R_1}
 \end{figure}

 Figure~\ref{fig:ite_EE} depicts the convergence behavior of Algorithm~\ref{alg:Framwork_all} under different values of the number of REs $M$ and the inter-element spacing $d_s$. As the iteration proceeds, the achievable rate $R$ increases rapidly and subsequently stabilizes, thereby confirming the effectiveness of the proposed algorithm and the validity of the theoretical analysis. It is also observed that the convergence rate is influenced by both $M$ and $d_s$. For a fixed physical size of the active RIS, reducing the spacing distance---corresponding to a larger number of reflection elements---leads to slower convergence. This is attributed to the fact that denser REs deployment intensifies MC effects, increasing the complexity of the optimization problem. Nevertheless, the achievable rate improves with an increasing number of REs, as the enhanced array gain compensates for the performance degradation caused by stronger MC. Moreover, it is observed that the performance disparity between the proposed EMC-Active RIS model with MC and the active RIS model with MC unawareness becomes increasingly pronounced as $d_s$ decreases. This trend arises because denser deployment of REs intensifies MC effects, which severely degrades performance when ignored. In contrast, the EMC-Active RIS model with MC effectively mitigates the adverse impact of MC, thereby preserving its performance advantage under closely spaced REs configurations.
  
  Figure~\ref{fig:ite_R_1} compares the proposed DA-based AO algorithm with the baseline method (the S-OPT algorithm in \cite{S_parameter})  in terms of convergence behavior and computational complexity uunder different values of the number of REs $M$ and the inter-element spacing $d_s$. As illustrated in Fig.~\ref{fig:ite_R_1}(a), the proposed algorithm converges significantly faster than the baseline, requiring several orders of magnitude fewer iterations to reach the same convergence criterion, and its iteration number exhibits only a mild increase as $M$ grows. In contrast, the S-OPT algorithm suffers from a rapidly increasing iteration burden with the RIS size. Fig.~\ref{fig:ite_R_1}(b) further reports the CPU computation time, where the DA-based AO algorithm consistently achieves over three orders-of-magnitude reduction compared with the baseline across all RIS configurations. The performance gap becomes more pronounced for larger $M$, indicating that the DA-based AO algorithm is more scalable and computationally efficient for large-scale RIS-assisted systems. 
  
  \begin{figure}[t]
  	\vspace{-15pt}
  	\centering
  	\includegraphics[scale=0.38]{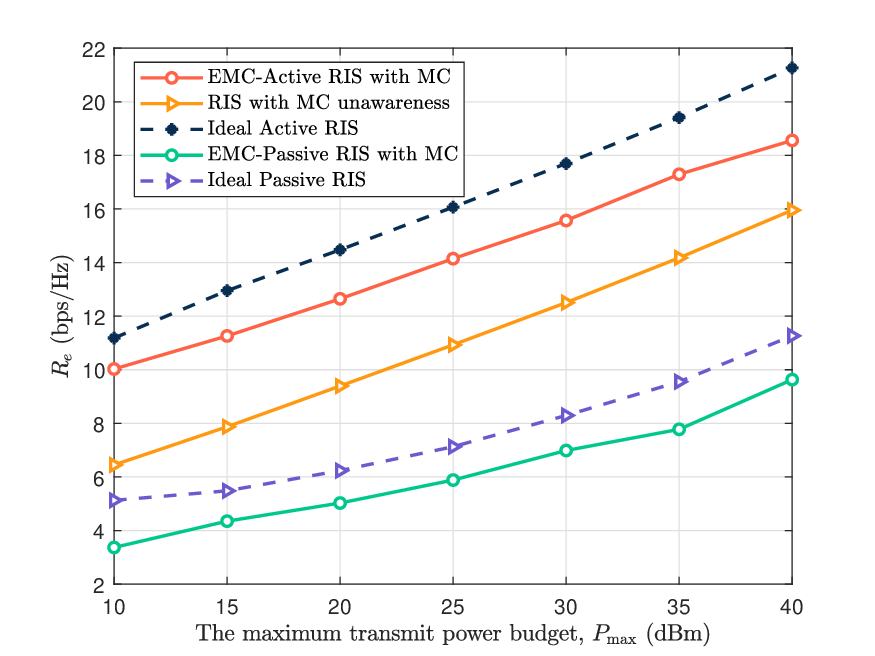}
  	\vspace{-10pt}
  	\caption{The ergodic achievable rate $R_e$ versus the maximum transmit power budget $P_{\max}$, where $d_s=\lambda_c/4$.}
  	\vspace{-13pt}
  	\label{fig:P_BS_EE}
  \end{figure}

\begin{figure}[t]
	\centering
	\includegraphics[scale=0.38]{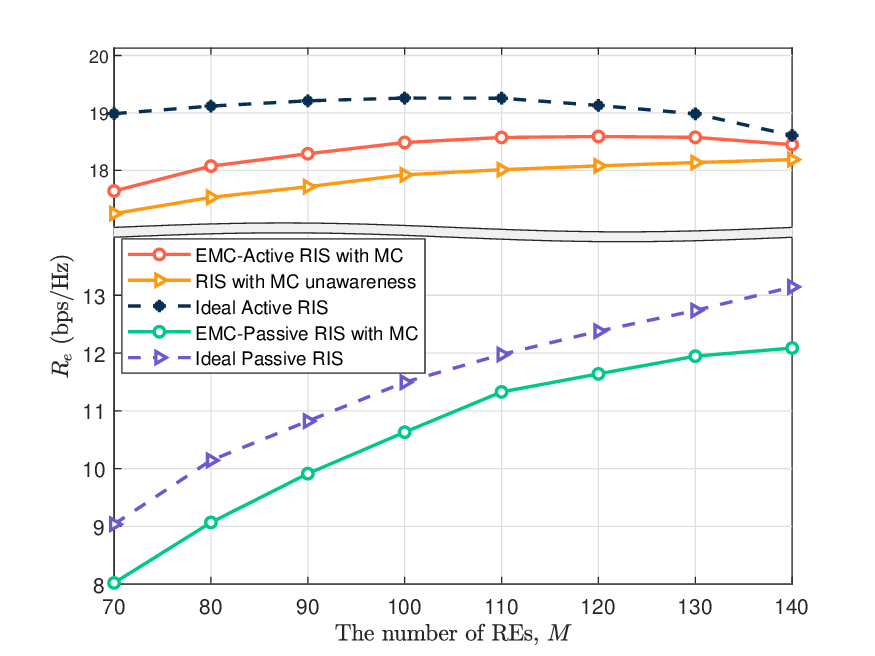}
	\vspace{-10pt}
	\caption{The ergodic achievable rate $R_e$ versus the number $M$ of REs, where $d_s=\lambda_c/4$.}
	\vspace{-13pt}
	\label{fig:M_R}
\end{figure}

  Figure~\ref{fig:P_BS_EE} illustrates the ergodic achievable rate $R_e$ versus the maximum transmit power budget $P_{\max}$ for different schemes. Although $R_e$ increases with $P_{\max}$ for all schemes, their growth behaviors differ noticeably. Owing to signal amplification at the RIS, active RIS-based schemes exhibit substantially higher rate gains than passive RIS as $P_{\max}$ increases. Moreover, the EMC-Active RIS with MC achieves performance close to that of the ideal active RIS, confirming that the EMC-Active RIS modeling preserves the performance benefits of active RIS. In contrast, neglecting MC results in suboptimal designs and a clear performance loss, as reflected by the gap between the EMC-Active RIS with MC scheme and the RIS with MC unawareness scheme. 
  
  Figure~\ref{fig:M_R} shows $R_e$ versus the number of REs $M$ with fixed inter-element spacing. Consistent with Fig.~\ref{fig:P_BS_EE}, similar performance gaps are observed among different schemes. As $M$ increases, the achievable rate of the active RIS first improves and then gradually degrades. This is because the amplification power budget $P_{A,\max}$ is sufficient to fully amplify a small number of elements, yielding increasing array gains, but becomes inadequate as more REs are added, leading to reduced per-RE amplification. This result reveals a fundamental trade-off between the amplification power budget and the reflection amplitude constraint, underscoring the importance of jointly selecting $P_{A,\max}$ and $\Gamma_{A,\max}$ in active RIS design.

   \begin{figure}[t]
   	\vspace{-15pt}
   	\centering
   	\includegraphics[scale=0.38]{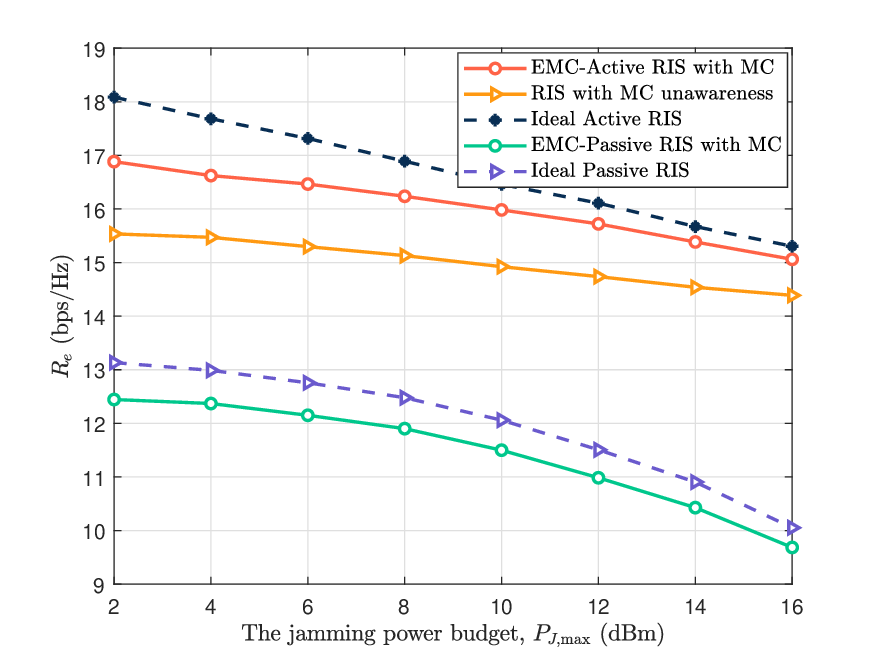}
   	\vspace{-10pt}
   	\caption{The ergodic achievable rate $R_e$ versus the jamming power budget $P_{J,\max}$, where $d_s=\lambda_c/4$.}
   	\vspace{-13pt}
   	\label{fig:PJ}
   \end{figure}

  \begin{figure}[t]
  	\centering
  	\includegraphics[scale=0.38]{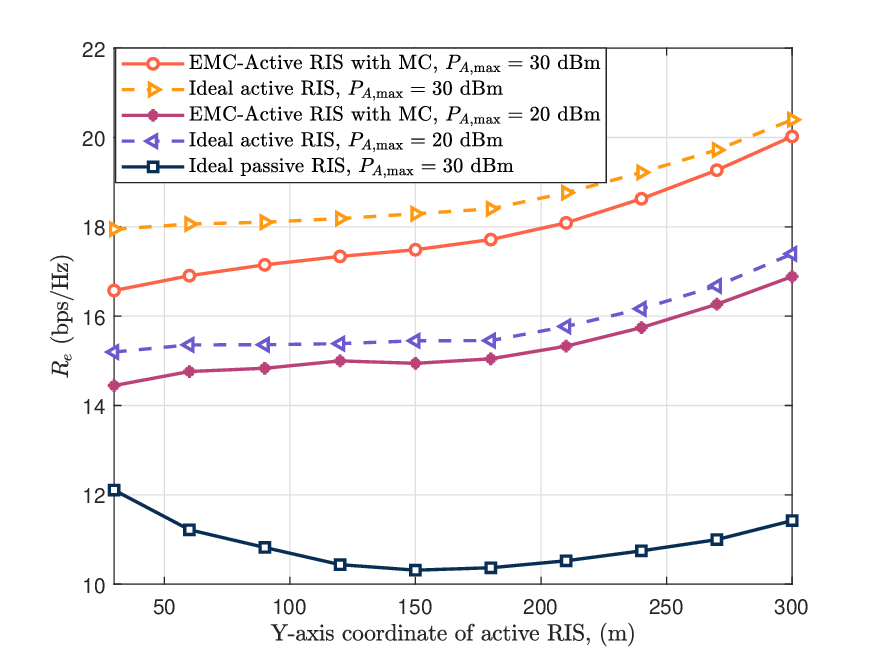}
  	\vspace{-10pt}
  	\caption{The ergodic achievable rate $R_e$ versus the Y-axis coordinate of active RIS, where $d_s=\lambda_c/4$.}
  	\vspace{-13pt}
  	\label{fig:L_EE}
  	\end{figure}
  
  Figure~\ref{fig:PJ} presents $R_e$ versus the jamming power budget $P_{J,\max}$ for different RIS schemes. As the jamming power increases, all schemes experience performance degradation, albeit at different rates. Notably, the EMC-Active RIS with MC scheme exhibits the most robust performance, with only a marginal rate reduction and behavior close to the ideal benchmark. These results demonstrate the superiority of the EMC-Active RIS model with MC in mitigating strong jamming effects under practical MC conditions.

  Figure~\ref{fig:L_EE} illustrates $R_e$ versus the RIS location, which varies with the Y-axis coordinate of RIS. It is observed that the passive RIS exhibits the poorest performance when placed near the midpoint between the BS and UEs, where the double-fading effect is most severe. In contrast, the ergodic achievable rate of the EMC-Active RIS-assisted system improves as the RIS moves closer to UEs. This behavior arises because weaker impinging signals at the active RIS can be more effectively amplified, leading to enhanced system performance. These results indicate that deploying the EMC-Active RIS in proximity to the UEs yields more pronounced performance gains.

  \vspace{-3mm}

\section{Conclusions}\label{sec:Conclusion}

 In this paper, we developed an EMC-Active RIS-assisted anti-jamming framework to jointly account for the EM and physical properties of active RIS, such as mutual coupling effects, channel correlation, and discrete phase. To assess the effectiveness of the proposed framework, a low-complexity DA-based AO algorithm was proposed to maximize the ergodic achievable rate. By employing the DA to mitigate MC among REs, the original coupled system was transformed into an equivalent uncoupled representation, enabling the non-convex optimization problem to be decomposed into three tractable subproblems. Simulation results demonstrate that the proposed DA-based AO algorithm substantially reduces the modeling and optimization complexity while achieving efficient convergence with significantly fewer iterations.

\begin{appendices}
\section{Proof of Theorem 1}
 
 To begin with, based on \cite[Eq.~(3)]{BD_RIS}, we have the following equations
 
 \vspace{-3mm}
 \begin{small}
 	\begin{align}\label{G_A1}
 	&\boldsymbol{\rm \Gamma}_{A}=\left(\boldsymbol{\rm Z}_{A}+Z_0\boldsymbol{\rm I}_{M}\right)^{-1}\left(\boldsymbol{\rm Z}_{A}-Z_0\boldsymbol{\rm I}_{M}\right),\\
 	&\boldsymbol{\rm S}_{AA}=\left(\boldsymbol{\rm Z}_{AA}+Z_0\boldsymbol{\rm I}_{M}\right)^{-1}\left(\boldsymbol{\rm Z}_{AA}-Z_0\boldsymbol{\rm I}_{M}\right).\label{S_AA1}
 	\end{align}
 \end{small}We define $\boldsymbol{\rm  Z}_{AA}^+ = \boldsymbol{\rm  Z}_{AA}+Z_0\boldsymbol{\rm  I}_{M}$, $\boldsymbol{\rm  Z}_{AA}^- = \boldsymbol{\rm  Z}_{AA}-Z_0\boldsymbol{\rm  I}_{M}$, $\boldsymbol{\rm  Z}_{A}^+ = \boldsymbol{\rm  Z}_{A}+Z_0\boldsymbol{\rm  I}_{M}$, and $\boldsymbol{\rm  Z}_{A}^- = \boldsymbol{\rm  Z}_{A}-Z_0\boldsymbol{\rm  I}_{M}$. Therefore, the term $\left(\boldsymbol{\rm I}_{M}-\boldsymbol{\rm \Gamma}_{A}\boldsymbol{\rm  S}_{AA}\right)^{-1}\boldsymbol{\rm \Gamma}_{A}$ in the equivalent channels $\boldsymbol{\rm H}^{E}_{S,k}$, $\boldsymbol{\rm H}^{J}_{S,qk}$, and $\boldsymbol{\rm H}^{N}_{S,k}$ can be rewritten as 

\vspace{-3mm}
\begin{small}
	\begin{align}\label{gam_1}
	&\hspace{-4.5mm}\left(\boldsymbol{\rm I}_{M}-\boldsymbol{\rm \Gamma}_{A}\boldsymbol{\rm  S}_{AA}\right)^{-1}\boldsymbol{\rm \Gamma}_{A}\notag\\
	&\overset{(a)}{=}\left(\boldsymbol{\rm I}_{M}-\left(\boldsymbol{\rm  Z}_{A}^{+}\right)^{-1}\boldsymbol{\rm  Z}_{A}^{-}\left(\boldsymbol{\rm  Z}_{AA}^{+}\right)^{-1}\boldsymbol{\rm  Z}_{AA}^{-}\right)^{-1}\left(\boldsymbol{\rm  Z}_{A}^{+}\right)^{-1}\boldsymbol{\rm  Z}_{A}^{-}\notag\\
	&\overset{(b)}{=}\left(\boldsymbol{\rm  Z}_{A}^{+}-\boldsymbol{\rm  Z}_{A}^{-}\left(\boldsymbol{\rm  Z}_{AA}^{+}\right)^{-1}\boldsymbol{\rm  Z}_{AA}^{-}\right)^{-1}\boldsymbol{\rm  Z}_{A}^{-},
	\end{align}
\end{small}where (a) holds due to substituting the Eqs.~(\ref{G_A1}) and (\ref{S_AA1}), and (b) holds because $\boldsymbol{\rm  Z}_{A}^{+}$ from (a) is moved outside the parentheses. For the term $\boldsymbol{\rm  Z}_{A}^{+}-\boldsymbol{\rm  Z}_{A}^{-}\left(\boldsymbol{\rm  Z}_{AA}^{+}\right)^{-1}\boldsymbol{\rm  Z}_{AA}^{-}$ of Eq.~(\ref{gam_1}), we have 

\vspace{-3mm}
\begin{small}
	\begin{align}\label{gam_2}
	\boldsymbol{\rm  Z}_{A}^{+}-\boldsymbol{\rm  Z}_{A}^{-}\left(\boldsymbol{\rm  Z}_{AA}^{+}\right)^{-1}\boldsymbol{\rm  Z}_{AA}^{-}
	\hspace{-0.5mm}&\overset{(a)}{=}\boldsymbol{\rm  Z}_{A}^{+}-\boldsymbol{\rm  Z}_{A}^{-}+2Z_0\boldsymbol{\rm  Z}_{A}^{-}\left(\boldsymbol{\rm  Z}_{AA}^{+}\right)^{-1}\notag\\
	&\overset{(b)}{=}2Z_0\hspace{-0.5mm}\left(\boldsymbol{\rm  I}_{M}+\boldsymbol{\rm  Z}_{A}^{-}\left(\boldsymbol{\rm  Z}_{AA}^{+}\right)^{-1}\hspace{-0.5mm}\right)\hspace{-1mm},
	\end{align}
\end{small}where (a) holds because $\boldsymbol{\rm  Z}_{AA}^{-}=\boldsymbol{\rm  Z}_{AA}^{+}-2Z_0\boldsymbol{\rm  I}_{M}$, and (b) holds due to substituting the definition of $\boldsymbol{\rm  Z}_{A}^{+}$ and $\boldsymbol{\rm  Z}_{A}^{-}$. Then, substituting Eq.~(\ref{gam_2}) into (\ref{gam_1}) and employing the push-through identity  yields 

\vspace{-3mm}
\begin{small}
	\begin{align}\label{gam_3}
	\left(\boldsymbol{\rm I}_{M}-\boldsymbol{\rm \Gamma}_{A}\boldsymbol{\rm  S}_{AA}\right)^{-1}\boldsymbol{\rm \Gamma}_{A}{=}\frac{1}{2Z_0}\boldsymbol{\rm  Z}_{A}^{-}\left(\boldsymbol{\rm  Z}_{A}+\boldsymbol{\rm  Z}_{AA}\right)^{-1}\boldsymbol{\rm  Z}_{AA}^{+}.
	\end{align}
\end{small}
 \vspace{-3mm}

 Afterwards, the equivalence relationship between $Z$-parameter matrices and $S$-parameter matrices in equivalent channels $\boldsymbol{\rm H}_{S,k}^{E}$ and $\boldsymbol{\rm H}_{Z,k}^{E}$ is given as follows:
 
 \vspace{-3mm}
 \begin{small}
 	\begin{align}\label{equ1}
 	&\boldsymbol{\rm S}_{BU,k}=\frac{\boldsymbol{\rm Z}_{BU,k}}{2Z_0}-\frac{1}{2Z_0}\boldsymbol{\rm Z}_{RU,k}\left(\boldsymbol{\rm  Z}_{AA}^{+}\right)^{-1}\boldsymbol{\rm Z}_{BR},\\
 	&\boldsymbol{\rm S}_{BR}=\left(\boldsymbol{\rm  Z}_{AA}^{+}\right)^{-1}\boldsymbol{\rm Z}_{BR},\label{equ2}\\
 	&\boldsymbol{\rm S}_{RU,k}=\frac{\boldsymbol{\rm Z}_{RU,k}}{2Z_0}\left(\boldsymbol{\rm I}_{M}-\left(\boldsymbol{\rm  Z}_{AA}^{+}\right)^{-1}\boldsymbol{\rm  Z}_{AA}^{-}\right).\label{equ3}
 	\end{align}
 \end{small}By inserting Eqs.~(\ref{equ2}) and (\ref{gam_3}) into (\ref{HE}) , we can get

\vspace{-3mm}
\begin{small}
	\begin{align}\label{HE_zh1}
	\boldsymbol{\rm H}_{S,k}^{E}=\boldsymbol{\rm S}_{BU,k}+\frac{1}{2Z_0}\boldsymbol{\rm  S}_{RU,k}\boldsymbol{\rm  Z}_{A}^{-}\left(\boldsymbol{\rm  Z}_{A}+\boldsymbol{\rm  Z}_{AA}\right)^{-1}\boldsymbol{\rm Z}_{BR}.
	\end{align}
\end{small}Then, substituting Eq.~(\ref{equ1}) and (\ref{equ3}) into (\ref{HE_zh1}) and simplifying the expression, we can obtain

\vspace{-3mm}
\begin{small}
	\begin{align}\label{HE_zh2}
	\boldsymbol{\rm H}_{Z,k}^{E}=&\frac{1}{2Z_0}\boldsymbol{\rm Z}_{BU,k}+\frac{1}{2Z_0}\boldsymbol{\rm  Z}_{RU,k}\left(\boldsymbol{\rm  Z}_{AA}^{+}\right)^{-1}\notag\\
	&\hspace{5mm}\times\left(\boldsymbol{\rm  Z}_{A}^{-}\left(\boldsymbol{\rm  Z}_{A}+\boldsymbol{\rm  Z}_{AA}\right)^{-1}-\boldsymbol{\rm  I}_{M}\right)\boldsymbol{\rm Z}_{BR}\notag\\
	=&\frac{1}{2Z_0}\left(\boldsymbol{\rm Z}_{BU,k}-\boldsymbol{\rm  Z}_{RU,k}\left(\boldsymbol{\rm Z}_{A}+\boldsymbol{\rm  Z}_{AA}\right)^{-1}\boldsymbol{\rm Z}_{BR}\right).
	\end{align}
\end{small}Similarly, by substituting the Eqs.~(\ref{gam_3})-(\ref{equ3}) into $\boldsymbol{\rm H}_{S,qk}^{J}$, $\boldsymbol{\rm H}_{S,k}^{N}$, $\boldsymbol{\rm \breve {H}}_{S}^{E}$, $\boldsymbol{\rm \breve {H}}_{S,q}^{J}$, and $\boldsymbol{\rm \breve {H}}_{S}^{N}$, we can derive the expressions of $\boldsymbol{\rm H}_{Z,qk}^{J}$, $\boldsymbol{\rm H}_{Z,k}^{N}$, $\boldsymbol{\rm \breve {H}}_{Z}^{E}$, $\boldsymbol{\rm \breve {H}}_{Z,q}^{J}$, and $\boldsymbol{\rm \breve {H}}_{S}^{N}$, respectively. Proof is complete.

 \section{Proof of Proposition 1}
 
 To begin with, we provide the following lemma, which plays a crucial role in this proof.
 \begin{lemma}\label{lemma1}
 	For a random matrix $\boldsymbol{\rm Z}$ following Gaussian distribution $\boldsymbol{\rm {Z}}\sim\mathcal{CN}\left(\boldsymbol{\rm \mu},\boldsymbol{\rm {\Sigma}}_{1}\otimes\boldsymbol{\rm {\Sigma}}_{2}\right)$ and given matrix $\boldsymbol{\rm {A}}$, $\mathbb{E}[\boldsymbol{\rm {Z}}^H\boldsymbol{\rm {A}}\boldsymbol{\rm {Z}}]$ can be expressed as 
 	
 	\vspace{-3mm}
 	\begin{small}
 		\begin{align}
 		\mathbb{E}\left[\boldsymbol{\rm {Z}}^H\boldsymbol{\rm {A}}\boldsymbol{\rm {Z}}\right]=\boldsymbol{\rm \mu}^H\boldsymbol{\rm {A}}\boldsymbol{\rm \mu}+\mbox{Tr}(\boldsymbol{\rm {A}}\boldsymbol{\rm {\Sigma}}_{1})\boldsymbol{\rm {\Sigma}}_{2}.
 		\end{align}
 	\end{small}
 \end{lemma}
 
 To handle the expectations in $\widetilde{R}_e$, we first derive the numerator part of $\gamma_{k}$ in $\widetilde{R}_e$, as follow:
 
 \vspace{-3mm}
 \begin{small}
 	\begin{align}
 	&\mathbb{E}\left[\left\vert\boldsymbol{\rm H}^{E}_{Z,k}\boldsymbol{\rm w}_k\right\vert^2\right]\overset{(a)}{=}\boldsymbol{\rm w}_k^H\mathbb{E}\left[\frac{1}{\left(2Z_0\right)^2}\boldsymbol{\rm Z}_{BU,k}^H\boldsymbol{\rm Z}_{BU,k}\right]\boldsymbol{\rm w}_k\notag\\
 	&\qquad+\boldsymbol{\rm w}_k^H\mathbb{E}\left[\frac{1}{\left(4Z_0\right)^2}\boldsymbol{\rm Z}_{BR}^H\boldsymbol{\rm \widetilde \Gamma}_{A}^H\boldsymbol{\rm \widetilde Z}_{RU,k}^H\boldsymbol{\rm \widetilde Z}_{RU,k}\boldsymbol{\rm \widetilde \Gamma}_{A}\boldsymbol{\rm Z}_{BR}\right]\boldsymbol{\rm w}_k\notag\\
 	&\overset{(b)}{=}\boldsymbol{\rm w}_k^H\boldsymbol{\rm R}_{BU,k}\boldsymbol{\rm w}_k+\boldsymbol{\rm w}_k^H\mathbb{E}\left[\frac{1}{\left(4Z_0\right)^2}\boldsymbol{\rm Z}_{BR}^H\boldsymbol{\rm \widetilde \Gamma}_{A}^H\boldsymbol{\rm \widetilde R}_{RU,k}\boldsymbol{\rm \widetilde \Gamma}_{A}\boldsymbol{\rm Z}_{BR}\right]\boldsymbol{\rm w}_k\notag\\
 	&\overset{(c)}{=}\boldsymbol{\rm w}_k^H\boldsymbol{\rm C}_{1,k}\boldsymbol{\rm w}_k,
 	\end{align}
 \end{small}where $(a)$ holds due to the substitution of expression for $\boldsymbol{\rm H}^{E}_{Z,k}$; $(b)$ holds due to the fact that $\boldsymbol{\rm R}_{BU,k}=\mathbb{E}\left[\frac{1}{\left(2Z_0\right)^2}\boldsymbol{\rm Z}_{BU,k}^H\boldsymbol{\rm Z}_{BU,k}\right]$, $\boldsymbol{\rm\widetilde {Z}}_{RU,k}\sim\mathcal{CN}\left(\boldsymbol{\rm \widetilde\mu}_{RU,k},\left(\varepsilon_{\varpi}^N\right)^2\boldsymbol{\rm {\widetilde\Sigma}}_{R}\right)$ and $\boldsymbol{\rm\widetilde R}_{RU,k}=\mathbb{E}\left[\boldsymbol{\rm\widetilde Z}_{RU,k}^H\boldsymbol{\rm\widetilde Z}_{RU,k}\right]$; $(c)$ stands due to applying the Lemma~\ref{lemma1}. By employing the method to the other expectations in the $\textbf{\textit{P}1}$, we can obtain the result in Proposition~\ref{proposition}.

\end{appendices}
 
\bibliographystyle{IEEEtran}
\bibliography{References}

\end{document}